\documentclass[11pt,letterpaper]{article}
\usepackage{etoolbox}
\usepackage{graphicx}
\usepackage{latexsym}
\usepackage{amssymb}
\usepackage{amsmath}
\usepackage{amsthm}
\usepackage{forloop}
\usepackage[paper=letterpaper,margin=1in]{geometry}
\usepackage[ruled,vlined,linesnumbered]{algorithm2e}
\usepackage{paralist}
\usepackage{appendix}
\usepackage{lmodern}

\newtheorem{theorem}{Theorem}[section]
\newtheorem{lemma}[theorem]{Lemma}

\newtheorem{corollary}[theorem]{Corollary}
\newtheorem{definition}{Definition}[section]
\newtheorem{proposition}[theorem]{Proposition}
\newtheorem{observation}[theorem]{Observation}
\newtheorem{claim}[theorem]{Claim}
\newtheorem{assumption}[theorem]{Assumption}

\newtheorem{reduction}{Reduction}

\newcommand{\defcal}[1]{\expandafter\newcommand\csname c#1\endcsname{{\mathcal{#1}}}}
\newcommand{\defbb}[1]{\expandafter\newcommand\csname b#1\endcsname{{\mathbb{#1}}}}
\newcounter{calBbCounter}
\forLoop{1}{26}{calBbCounter}{
    \edef\letter{\Alph{calBbCounter}}
    \expandafter\defcal\letter
		\expandafter\defbb\letter
}

\newcommand{\ie}{{\it i.e.}}
\newcommand{\eg}{{\it e.g.}}

\newcommand{\DSS}[2]{{\cD_{#1}^+(#2)}}
\newcommand{\DS}[1]{{\cD^+(#1)}}
\newcommand{\DO}[1]{{D^*(#1)}}
\newcommand{\DSF}[1]{{\cD^+_{#1}}}

\newcommand{\maxm}{{\mathsf{m}_{\max}}}
\newcommand{\maxs}{{\mathsf{S}_{\max}}}
\newcommand{\opt}{{\mathsf{opt}}}
\DeclareMathOperator{\Var}{Var}
\DeclareMathOperator{\Cov}{Cov}
\newcommand{\poly}{{\mathtt{Poly}}}

\begin{document}

\newcommand {\ignore} [1] {}
\newcommand \xor{\oplus}
\newcommand \NC{{\cal NC}}
\newcommand \NP{{\cal NP}}
\newcommand \SSE{{SSE}}
\newcommand \APX{{\cal APX}}
\newcommand{\noticeable}{\mathop{\mathrm{noticeable}}\nolimits}
\newcommand{\negligible}{\mathop{\mathrm{negligible}}\nolimits}
\newcommand{\size}{\mathop{\mathrm{size}}\nolimits}
\newcommand{\argmax}{\mathop{\mathrm{argmax}}\nolimits}
\newcommand{\prob}{\mathop{\mathrm{Pr}}\nolimits}
\newcommand{\e}{\varepsilon}
\newcommand{\II}{{\cal I}}
\newcommand{\PPP}{{\cal P}}
\newcommand{\AAA}{{\cal A}}
\newcommand{\CCC}{{\cal C}}
\newcommand{\eqdef}{\stackrel{\mbox{\tiny{def}}}{=}}

\newcommand{\NN}{\mathbb{N}}
\newcommand{\RR}{\mathbb{R}}
\newcommand{\LL}{{\cal L}}
\newcommand{\Poly}{{\mathtt{Poly}}}
\newcommand{\nonnegR}{\mathbb{R}^+}

\title{\textbf{Building a Good Team: Secretary Problems and the Supermodular Degree}}

\author{
Moran Feldman\thanks{School of Computer and Communication Sciences, EPFL, Switzerland.
Email: \texttt{moran.feldman@epfl.ch}.}
\and
Rani Izsak\thanks{Department of Computer Science and Applied Mathematics, Weizmann Institute of Science, Israel.
Email: \texttt{ran.izsak@weizmann.ac.il}.}
}

\date{}

\maketitle

\begin{abstract}
In the (classical) \textsc{Secretary Problem}, one has to hire the best among $n$ candidates. The candidates are interviewed, one at a time, at a uniformly random order, 
and one has to decide on the spot, whether to hire a candidate or continue interviewing. It is well known that the best candidate can be hired with a probability of $1/e$ (Dynkin, 1963). 
Recent works extend this problem to settings in which multiple candidates can be hired, subject to some constraint. Here, one wishes to hire a set of candidates maximizing a given objective set function. 

Almost all extensions considered in the literature assume the objective set function is either linear or submodular. 
Unfortunately, real world functions might not have either of these properties. Consider, for example, a scenario where one hires researchers for a project. 
Indeed, it can be that some researchers can substitute others for that matter.
However, it can also be that some combinations of researchers result in synergy 
(see, \eg, Woolley et al., Science 2010, for a research about collective intelligence).
The first phenomenon can be modeled by a submoudlar set function, while the latter cannot.

In this work, we study the secretary problem with an {\em arbitrary} non-negative monotone valuation function, subject to a general matroid constraint. 
It is not difficult to prove that, generally, only very poor results can be obtained for this class of objective functions. 
We tackle this hardness by combining the following: 
(1) Parametrizing our algorithms by the {\em supermodular degree} of the objective function (defined by Feige and Izsak, ITCS 2013), which, roughly speaking, measures the distance of a function from being submodular. 
(2) Suggesting an (arguably) natural model that permits approximation guarantees that are {\em polynomial} in the supermodular degree (as opposed to the standard model which allows only {\em exponential} guarantees). 
Our algorithms learn the input by running a non-trivial estimation algorithm on a portion of it whose size depends on the supermodular degree.

We also provide better approximation guarantees for the special case of a uniform matroid constraint.
To the best of our knowledge, our results represent the first algorithms for a secretary problem handling arbitrary non-negative monotone valuation functions.
\end{abstract}

\pagenumbering{gobble}
\newpage
\pagenumbering{arabic}

\vspace{-0.1in}
\section{Introduction} \label{sec:introduction}
\vspace{-0.1in}
In the (classical) \textsc{Secretary Problem}, one has to hire a worker from a pool of $n$ candidates. The candidates arrive to an interview at a uniformly random order, and the algorithm must decide immediately and irrevocably, after interviewing a candidate, whether to hire him or continue interviewing. The objective is to hire the best candidate. It is well-known that the best candidate can be hired with a probability of $1/e$, and that this is asymptotically optimal~\cite{Dynkin}.

Recently, there has been an increased interest in variants of the secretary problem where more than a single candidate can be selected, subject to some constraint (\eg, a matroid constraint). Such variants have important applications in mechanism design 
(see, \eg,~\cite{AKW14,BIKK07,BIKK08,K05} and the references therein). 
When more than one candidate can be selected, there is a meaning to the values of {\em subsets} of candidates.
If one allows these values to be determined by an arbitrary non-negative monotone set function, then only exponentially competitive ratios (in the number of candidates) can be achieved, even subject to a simple cardinality constraint.%
\footnote{Intuitively, the bad example consists of a cardinality constraint allowing us to select only $k$ candidates, and an objective function assigning a strictly positive value only to sets containing $k$ specific candidates or more than $k$ candidates. In this case the algorithm has no room for mistakes, which leads to a very poor performance.}

In light of the above hardness, previous works have concentrated on restricted families of objective functions, such as linear and submodular functions
(see, \eg, \cite{BIK07,BHZ13,CL12,DP12,FSZ15,GRST10,L14} and a more thorough discussion in Section~\ref{sec:relatedWork}).
However, for many applications, the desired set function might admit complements, \ie, a group of candidates might exhibit synergy and contribute more as a group than the sum of the candidates' personal contributions (see also Woolley et al.~\cite{WCPHM10} for a research about collective intelligence). 
Such complements cannot be modeled by submodular (or linear) objectives. Dealing with complements in general results in unacceptable guarantees, as discussed above. 
However, what if one has a function which is submodular, except for a pair of candidates which are better to hire together? Can we guarantee anything for this case?

In this paper, we give a strong affirmative answer to this question.
Specifically, we give algorithms for secretary problems with {\em arbitrary} non-negative monotone objective functions, whose guarantees are proportional 
to the distance of the objective function from being submodular, as measured by the {\em supermodular degree} (defined by~\cite{FI13}).
Back to the above example, the pair of synergistic candidates results in an objective function with a supermodular degree of~1, and for such an objective our algorithms provide a constant competitive ratio for the problem of hiring a team of a given size.
Intuitively, the supermodular degree can be seen as measuring the number of candidates that any single candidate can have synergy with.
Our algorithms handle both the case of a cardinality constraint (demonstrated above) and the more general case of a matroid constraint. 
For a cardinality constraint, we obtain a constant competitive ratio when the supermodular degree is constant. 
For a (general) matroid constraint, our competitive ratios depend logarithmically on the rank of the matroid.%
\footnote{Note that, till a very short while ago, this was the case even for the state of the art algorithm for a {\em submodular} objective function (which corresponds to a supermodular degree of~0)~\cite{GRST10}. An improved algorithm with a competitive ratio of $O(\log \log k)$ (where $k$ is the rank of the matroid) has been recently given by~\cite{FZ15}.}
To the best of our knowledge, these are the first algorithms for the secretary problem with an arbitrary non-negative monotone objective set function.

\vspace{-0.1in}
\section{Preliminaries} \label{sec:preliminaries}
\vspace{-0.1in}

For completeness of the presentation, we give in this section a few relevant definitions from the literature (see, \eg, \cite{FeldmanIzsak14}). 
All set functions in this work are non-negative and (non-decreasing) monotone\footnote{A set function $f\colon 2^\cN \to \nonnegR$ is monotone if and only if $f(S) \leq f(T)$ whenever $S \subseteq T \subseteq \cN$.}. For readability, given a set $S \subseteq \cN$ and an element $u \in \cN$ we use $S + u$ to denote $S \cup \{u\}$ and $S - u$ to denote $S \setminus \{u\}$. 

\subsection{Matroids}
Given a ground set $\cN$, a pair $(\cN, \cI)$ is called a \textsf{matroid} if $\cI \subseteq 2^\cN$ obeys three properties:
\begin{compactenum}[(i)]
	\item $\cI$ is non-empty.
	\item $\cI$ is hereditary, \ie, $S \subseteq T \subseteq \cN$ and $T \in \cI$ imply $S \in \cI$.
	\item For every two sets $S, T \in \cI$ such that $|S| > |T|$, there exists an element $u \in S \setminus T$, such that $T + u \in \cI$. This property is called the \emph{augmentation property} of matroids.
\end{compactenum}

We say that a set $S \subseteq \cN$ is \emph{independent} if $S \in \cI$. The rank of a matroid is the size of the largest independent set in $\cI$. One important class of matroids, which is central to our work, is the class of \emph{uniform} matroids. A uniform matroid of rank $k$ is simply a cardinality constraint, \ie, a set $S \subseteq \cN$ is independent in such a matroid if and only if its cardinality is at most $k$.

\subsection{Supermodular degree}
The following standard definition is very handy.

\begin{definition} [Marginal set function] \label{def:marginal}
Let $f:2^{\cN} \to \nonnegR$ be a set function and let $u \in \cN$.
The \textsf{marginal set function} of $f$ with respect to $u$, denoted by $f(u \mid \cdot)$
is defined as $f(u \mid S) \eqdef f(S + u) - f(S)$.
When the underlying set function $f$ is clear from the context, we sometimes call $f(u \mid S)$ the \textsf{marginal contribution} of $u$ to the set $S$.
Similarly, for subsets $S, T \subseteq \cN$,
we use the notation $f(T \mid S) \eqdef f(S \cup T) - f(S)$.
\end{definition}

We can now give the definition of the supermodular degree (originally defined by~\cite{FI13}), which is used to parameterize our results.

\begin{definition} [Supermodular (dependency) degree] \label{def:supermodular}
The \textsf{supermodular degree of an element} $u \in \cN$ with respect to $f$ is defined as the cardinality of the set $\DSS{f}{u} = \{v \in \cN \mid \exists_{S \subseteq \cN} f(u \mid S + v) > f(u \mid S)\}$, containing all elements whose existence in a set might \emph{increase} the marginal contribution of $u$. $\DSS{f}{u}$ is called the \textsf{supermodular dependency set} of $u$ with respect to $f$, and we sometimes refer to the elements of $\DSS{f}{u}$ as \textsf{supermodular dependencies}.
The \textsf{supermodular degree of a function} $f$, denoted by $\DSF{f}$, is simply the maximum supermodular degree of any element $u \in \cN$. Formally, $\DSF{f} = \max_{u \in \cN} | \DSS{f}{u} |$.
When the underlying set function is clear from the context, we sometimes omit it from the notations.
\end{definition}

Note that $0 \leq \DSF{f} \leq n-1$ for any set function $f$. More specifically, $\DSF{f} = 0$ when $f$ is {\em submodular}, and becomes larger as $f$ deviates from submodularity. When the function $f$ is clear from the context, we use $d$ to denote $\DSF{f}$.

\subsection{Input representation}

In general, a set function might assign $2^n$ different values for the subsets of a ground set of size $n$.
Thus, not every set function has a succinct (\ie, polynomial in $n$) representation.
Therefore, it is a common practice to assume access to a set function via an oracle.
That is, an algorithm handling a set function often gets access to an oracle that answers queries about the function, instead of getting an explicit representation of the function. Arguably, the most basic type of an oracle is the \textsf{value oracle}, which, given any subset of the ground set, returns the value assigned to it by the set function. 
Formally:
\begin{definition} 
\textsf{Value oracle} of a set function $f:2^{\cN} \to \nonnegR$ is the following:
\\\textsf{Input:} A subset $S \subseteq \cN$.
\\\textsf{Output:} $f(S)$.
\end{definition}

Similarly, since in a given matroid the number of independent subsets might be, in general, exponential in the size of the ground set,
it is common to assume access to the following type of oracle.
\begin{definition} 
\textsf{Independence oracle} of a matroid $(\cN, \cI)$ is the following:
\\\textsf{Input:} A subset $S \subseteq \cN$.
\\\textsf{Output:} A Boolean value indicating whether $S \in \cI$.
\end{definition}

Additionally, in order to manipulate a function with respect to the supermodular degree,
one needs a way to determine the supermodular dependencies of a given element of the ground set.
An oracle for that purpose was introduced by~\cite{FI13}, and was later used also by~\cite{FeldmanIzsak14}.
Formally:
\begin{definition}
\textsf{Supermodular oracle} of a set function $f:2^{\cN} \to \nonnegR$ is the following:
\\\textsf{Input:} An element $u \in \cN$.
\\\textsf{Output:} The set $\DS{u}$ of the supermodular dependencies of $u$ with respect to $f$.
\end{definition}

The above oracles definitions are acceptable for offline algorithms. In online settings, these oracles have to be weakened and limited to return only information that we ``expect'' the algorithm to have. Some possible weakened versions can be found in previous work. Still, finding a set of weakened oracles that ``makes sense'' in the context of the supermodular degree is not trivial. Our model, including the weakened oracles that we use, appears in Section~\ref{sec:model}.

\subsection{Online algorithms}
Like standard {\em online} algorithms, the performance of a secretary algorithm is measured by the \textsf{competitive ratio}, which is the worst case ratio between the expected performance of the algorithm and the performance of an offline optimal algorithm. More formally, let $\cP$ be the set of possible instances, $OPT(P)$ be the value of the optimal solution for an instance $P \in \cP$ and $ALG(P)$ be the value of the algorithm's solution given the instance $P$. Then, the competitive ratio (for maximization problems) of the algorithm is given by:
\[
	\sup_{P \in \cP} \frac{OPT(P)}{\bE[ALG(P)]}
	\enspace,
\]
where the expectation is over the randomness of the algorithm and the arrival order of the input.

\subsection{Techniques}

Most algorithms for secretary problems start with a learning phase in which they reject all elements, and later, after accumulating some information about the input, they move to a phase in which they may accept elements. When the value of an element might positively depend on other $d$ elements, there might be a set of $d + 1$ elements such that every reasonable solution must contain this set. In this case, any reasonably good algorithm must terminate the learning phase, with a significant probability, before any element of this set arrives.

This means that the learning phase of our algorithms consists of only about $1/d$ of the input (where $d$ is the supermodular degree of the objective function), and thus, they rarely see all the dependencies of an element in the learning phase. Hence, by the end of the learning phase, our algorithms cannot calculate an optimal solution for the sub-problem represented by the part of the input seen thus far. However, we show that it is possible to {\em estimate} the value of the optimal solution based on the learning phase, and this estimation is crucial for the performance of our algorithms.

\vspace{-0.1in}
\section{Model and results} \label{sec:model}
\vspace{-0.1in}

Consider the following scenario.
A client enters a store, and wants to buy herself a new phone.
However, the client's main motive to buy this phone is a novel accessory which is not supported by her old phone.
Thus, the client asks the salesman to buy the phone bundled with the accessory.
Unfortunately, the accessory is not available at that time, because supply does not meet the overwhelming demand.
If the client insists on buying the phone only bundled with the accessory, then the salesman can offer her to buy a phone now and get the accessory next week when a new supply shipment arrives.

On the other hand, consider a slightly different scenario.
In this scenario, the client tells the salesman that she wants the phone together with some accessory, but does not tell him which accessory it is. The client then offers the following deal: the salesman will give her the phone now, and the client will pay when the unspecified accessory becomes available. Clearly no salesman can accept such an offer.

Our model (below) assumes the more realistic first scenario. That is, in case of complementarity, the ``bidder'' is required to announce the future elements she needs in order to get the maximum value from the current ``product''.

Formally, an instance of the monotone matroid secretary problem consists of a ground set $\cN$ of size $n$, a non-negative monotone set function $f: 2^{\cN} \to \nonnegR$ and a matroid $M = (\cN, \cI)$. The execution of an algorithm for this problem consists of $n$ \textsf{steps} (also referred to as \textsf{times}). In each step the following occurs:
\begin{compactitem}
\item One element of $\cN$ is \textsf{revealed} (arrives), at a uniformly random order (without repetitions).
\item The algorithm must decide whether to include the element in its output (irrevocably).
\end{compactitem}
The objective of the algorithm is to select an independent subset maximizing $f$. When making decisions, the algorithm has access to the value of $n$, the supermodular degree of $f$ and the following oracles. The first oracle is the independence oracle given above, which gives information about the constraint. The second oracle gives information about the objective function. This oracle is the counterpart of the value oracle defined above.


\begin{definition}
\textsf{Online marginal oracle} of a set function $f:2^{\cN} \to \nonnegR$ is the following:
\\\textsf{Input:} An element $u \in \cN$ {\em that has already been revealed} and a subset $S \subseteq \cN$.
\\\textsf{Output:} $f(u \mid S)$.
\end{definition}
Note that $S$ does not have to be fully (or even partially) revealed, but $u$ does have to. That is, we can only ask marginal queries for elements that have been revealed, but we can ask for their marginal value with respect to any subset. Our algorithms use the online marginal oracle only for the purpose of finding the best marginal that an element $u$ can have with respect to already accepted elements and subsets of $\DS{u}$. This use is consistent with the above motivation of our model.

The last oracle of our model returns the dependency sets of already revealed elements.

\begin{definition}
\textsf{Online supermodular oracle} of a set function $f:2^{\cN} \to \nonnegR$ is the following:
\\\textsf{Input:} An element $u \in \cN$ {\em that has already been revealed}.
\\\textsf{Output:} $\DSF{f}(u)$.
\end{definition}

Observe that the last oracle can, in fact, be implemented using the online marginal oracle, albeit using an exponential time complexity in $n$. However, this oracle is a natural online variant of the supermodular oracle, and thus, it emphasizes the relation between our model and previous work on the supermodular oracle. If time complexity is not a priority, as is often the case when analyzing online algorithms, then every use of this oracle can be replaced by an appropriate procedure using the online marginal oracle.

\paragraph{Additional Notation.} In the rest of this paper we use the notation $\cN_u$ to denote the set of elements revealed up to the point in which a given element $u \in \cN$ is revealed (\ie, $\cN_u$ contains $u$ and every other element revealed before $u$).

\subsection{Our results}
In this section we formally state our results for the model introduced above.
Our first result is for instances where the matroid $M$ is of rank $k \leq \DSF{f} + 1$. This setting is interesting for two reasons: it is closely related to the classical secretary problem, and our algorithm for it is used as a building block in our algorithms for other settings. 

\begin{theorem} \label{thm:smallRankMatroid}
There exists an $O(k \log k)=O(\DSF{f} \log \DSF{f})$-competitive algorithm for the monotone matroid secretary problem when the rank of the matroid constraint $M$ is $k \leq \DSF{f} + 1$.
\end{theorem}

The time complexity of the above algorithm (and all our other algorithms) is $\poly(n, 2^{\DSF{f}})$. The exponential dependence of the time complexity on $\DSF{f}$ is unavoidable even for {\em offline} algorithms and a uniform matroid constraint (see~\cite{FeldmanIzsak14} for more details). Our main result is given by the next theorem.

\begin{theorem} \label{t:general_nonpoly}
There exists an $O(\DSF{f}^3 \log \DSF{f} + \DSF{f}^2 \log k)$-competitive algorithm for the monotone matroid secretary problem.
\end{theorem}

Interestingly, the above competitive ratio matches the, till recently, state of the art ratio of $O(\log k)$ by Gupta et al.~\cite{GRST10} for the case where $f$ is a submodular function, and extends it to any constant supermodular degree. For uniform matroids we have the following improved guarantee.

\begin{theorem} \label{t:uniform_nonpoly}
There exists an $O(\DSF{f}^3 \log \DSF{f})$-competitive algorithm for the monotone matroid secretary problem when the matroid $M$ is uniform.
\end{theorem}

Note that the guarantee of Theorem~\ref{t:uniform_nonpoly} has no dependence on $k$, and thus, it yields a constant competitive ratio for a constant supermodular degree.

It is handy to assume that $f$ is normalized (\ie, $f(\varnothing) = 0$). 
Reduction~\ref{r:normalized} in Appendix~\ref{app:normalizedReduction} shows that this assumption is without loss of generality, and thus, we implicitly assume it in all our proofs.

\paragraph{Lower Bounds.}
Note that even the {\em offline} version of maximizing a function $f$ with respect to matroid constraint is $\NP$-hard to approximate within a guarantee of $\Omega(\ln {\DSF{f}} / \DSF{f})$
(see, \eg,~\cite{FeldmanIzsak14, HSS06}).
Moreover, for a uniform matroid constraint, it is $\SSE$-hard to achieve any constant\footnote{That is, a guarantee that does not depend on $\DSF{f}$.} approximation guarantee (even) in the offline setting~\cite{FeldmanIzsak14}.

\subsection{Related results} \label{sec:relatedWork}

\paragraph{Secretary problem.}
Many variants of the secretary problem have been considered throughout the years, and we mention here only those most relevant to this work. Under a cardinality constraint of $k$, Babaioff, Immorlica, Kempe and Kleinberg~\cite{BIKK07} and Kleinberg~\cite{K05} achieve two incomparable competitive ratios of $e$ and $1/[1-O(1/\sqrt{k})]$, respectively, for linear objective functions. For submodular objective functions, the best algorithms have a competitive ratio of $8e^2 \approx 59$ for the general case~\cite{BHZ13} and a competitive ratio of  $(e^2 + e)/(e - 1) \approx 5.88$ when the objective is also monotone~\cite{FNS11c}.

The matroid secretary problem considers a linear objective and a general matroid constraint. This variant was introduced by Babaioff, Immorlica and Kleinberg~\cite{BIK07}, who described an $O(\log k)$-competitive algorithm for it (where $k$ is the rank of the matroid) and conjectured the existence of an $O(1)$-competitive algorithm. Motivated by this conjecture, $O(1)$-competitive algorithms have been obtained for a wide variety of special classes of matroids including graphic matroids~\cite{BIK07,KP09}, transversal matroids~\cite{BIK07,DP12,KP09}, co-graphic matroids~\cite{S13b}, linear matroids with at most $k$ non-zero entries per column~\cite{S13b}, laminar matroids~\cite{IW11,JSZ13,MTW13}, regular matroids~\cite{DK14}, and some types of decomposable matroids, including max-flow min-cut matroids~\cite{DK14}. However, progress on the general case has been much slower. An $O(\sqrt{\log k})$-competitive algorithm was described by Chakraborty and Lachish~\cite{CL12}, and very recently two $O(\log \log k)$-competitive algorithms were given by Lachish~\cite{L14} and Feldman, Svensson and Zenklusen~\cite{FSZ15}.

The submodular variant of the matroid secretary problem was also considered. For general matroids~\cite{GRST10} gave an $O(\log k)$-competitive algorithm, and $O(1)$-competitive algorithms were described for special classes of matroids including partition matriods~\cite{BHZ13,FNS11c,GRST10} and transversal and laminar matroids~\cite{MTW13}. A recent work~\cite{FZ15} shows that any algorithm for the linear variant can be translated, with a limited loss in the competitive ratio, into an algorithm for the submodular variant. This implies an $O(1)$-competitive algorithm for the submodular variant under any class of matroids admiting such an algorithm for linear objectives, and an $O(\log \log k)$-competitive algorithm for general matroids. A secretary problem with an even more general family of objective functions was considerd by Bateni, Hajiaghayi and Zadimoghaddam~\cite{BHZ13} who proved an hardness result for the family of subadditive objective functions. Finally, variants of the matroid secretary problem which use a different arrival process or a non-adversarial assignment of element values were also considered~\cite{OvV13,JSZ13,S13b}.

\paragraph{Complexity measures of set functions.}
Complexity measures of set functions have been previously studied.
Abraham, Babaioff, Dughmi and Roughgarden \cite{ABDR12} studied the welfare maximization problem with respect to a complexity measure that gives greater values to set functions (\ie, mark them as more ``complex'') as their complementarity increases, in some sense.
This complexity measure is applicable only to a restricted class of set functions.
The notion of supermodular degree, which we use in this work, was introduced by~\cite{FI13}, again, with an application to the welfare maximization problem, and was later used for a more general appication by~\cite{FeldmanIzsak14}.
A stronger complexity measure for set functions was studied by \cite{MPH}. However, their results assume access to a demand oracle (see Blumrosen and Nisan~\cite{DemandQueries09}),
which is common in the context of combinatorial auctions, but, to the best of our knowledge, has not been used outside this world. 

\paragraph{Complements in online settings.}
A different online setting exhibiting complements can be found in the work of Emek, Halld\'{o}rsson, Mansour, Patt-Shamir, Radhakrishnan and Rawitz on online set packing~\cite{OSP2012}.

\vspace{-0.1in}
\section{Small rank matroids (Theorem~\ref{thm:smallRankMatroid})} \label{sec:small_rank}
\vspace{-0.1in}

In this section, we describe the main intuitive ideas behind the proof of Theorem~\ref{thm:smallRankMatroid}. 
The proof itself is deferred to Appendix~\ref{app:small_rank}. 
For simplicity, we assume in this section that $M$ is a uniform matroid of rank $d + 1$. 
The extension to general matroids and smaller ranks is quite straightforward.

A natural generalization of the algorithm for the classical secretary problem is the following algorithm. First, during the learning phase, reject the first $O(n/d)$ elements. From the remaining elements, take the first one whose marginal contribution with respect to some of its future supermodular dependencies (\ie, the elements that may increase its marginal contribution and the algorithm can still choose to take) is better than any such contribution inspected thus far.

It is not difficult to argue that the best marginal contribution seen by the above algorithm is always at least $f(OPT) / (d + 1)$. However, to get a competitive ratio guarantee, we need to show that the algorithm manages to pick this best contribution with a significant probability. One approach to proving this claim is by generalizing the analysis of the classical secretary algorithm to this more general algorithm. Such a generalization requires lower bounding the probability that the following two events occur (at the same time).

\begin{compactitem}
\item The element with best marginal contribution arrives after the learning phase.
\item The second best marginal contribution, up to the point where we see the best contribution, is seen during the learning phase.
\end{compactitem}

Unfortunately, it is difficult to bound the above probability due to the following phenomenon. The earlier an element arrives, the more future supermodular dependencies it has, and thus, the higher its corresponding marginal contribution. Hence, elements in the learning phase tend to have larger marginal contributions in comparison to elements appearing after the learning phase.

To overcome this issue, we modify the algorithm. Specifically, instead of comparing the marginal contribution of the current element to the marginal contributions seen thus far, we compare it to the marginal contributions that the elements seen thus far could have if they would have arrived at this time (instead of the time in which they have really arrived). A similar idea has been previously used by a work on the case of a submodular objective function~\cite{FNS11c}.

Additionally, to get the exact approximation ratio guaranteed by Theorem~\ref{thm:smallRankMatroid}, the algorithm has to use a random threshold from a logarithmic scale. This allows the analysis to assume that (with a significant probability) the learning phase takes about half of the time up to the point when the best contribution is observed by the algorithm.

\section{Estimation aided algorithms} \label{sec:algorithmsWithEstimation}

We say that a value $\opt_\alpha$ is an $\alpha$-estimation of an optimum solution $OPT$ if it obeys
$
	f(OPT) / \alpha
	\leq
	\opt_\alpha
	\leq
	f(OPT)
$.
We say that an algorithm is $\alpha$-\emph{aided} if it assumes getting an $\alpha$-estimation of the optimum as part of its input. 
In this section we describe an aided algorithm for the case of a general matroid constraint. An improved aided algorithm for the special case of a uniform matroid constraint can be found in Appendix~\ref{app:Aided_UniformMatroid}. 
In the next section we explain how to convert our aided algorithms into non-aided ones. 
Note that our aided algorithms work even under a model where the arrival order is determined by an adversary.
However, the randomness of the input is required for converting them into non-aided algorithms.

\begin{theorem} \label{thm:general_aided}
For every $\alpha \geq 1$, there exists an $\alpha$-aided $O(d^2 \log (\alpha k))$-competitive algorithm for the monotone matroid secretary problem.
\end{theorem}

The algorithm we use to prove Theorem~\ref{thm:general_aided} is Algorithm~\ref{alg:general_aided}. A few of the ideas we use in Algorithm~\ref{alg:general_aided} and its analysis can be traced back to~\cite{BIK07}.

\begin{algorithm}[ht]
\caption{\textsf{$\alpha$-Aided Algorithm for General Matroids}} \label{alg:general_aided}
Let $p$ be a uniformly random integer from the set $\{-\lceil \log_2 k \rceil - 3, -\lceil \log_2 k \rceil - 2, \dotsc, \lceil \log_2 \alpha \rceil\}$.\\
Let $\tau \gets 2^p \cdot \frac{\opt_\alpha}{2}$.\\
Let $S \gets \varnothing$.\\
\For{every arriving element $u$}
{
	\If{there exits a set $\DO{u} \subseteq \DS{u} \setminus \cN_u$ such that $f(u \mid \DO{u} \cup S) \geq \tau$ and\\
	\hspace{2.5in} $S \cup \DO{u} + u \in \cI$}
	{
		Add $\DO{u} + u$ to $S$.
	}
}
\Return{$S$}.
\end{algorithm}

We define a weight $w(u)$ for every element $u \in OPT$ as follows: $w(u) = f(u \mid OPT \setminus \cN_u)$. For ease of notation, we extend $w$ to subsets of $OPT$ in the natural way. Let us denote $p_1 = -\lceil \log_2 k \rceil$ and $p_2 = \lceil \log_2 \alpha \rceil$. For every integer $p_1 \leq p \leq p_2$, we define a set (bucket) $OPT_p = \{u \in OPT \mid 2^p \cdot \frac{\opt_\alpha}{2} \leq w(u) \leq 2^p \cdot \opt_\alpha\}$.
Intuitively speaking, the following lemma shows that there is sufficient value in all the buckets together. The lemma holds since every element that does not get into any bucket must have a very low weight.

\begin{lemma} \label{lem:sum_weights}
$w\left(\bigcup_{p = p_1}^{p_2} OPT_p\right) \geq \frac{f(OPT)}{2}$.
\end{lemma}

We defer the proof of Lemma~\ref{lem:sum_weights} and the other lemmata of this section to Appendix~\ref{app:missing_proofs}. Our next objective is to show that if Algorithm~\ref{alg:general_aided} selects a value $p$, then its gain is proportional to $w(OPT_{p + 3})$. Whenever $S$ appears below it denotes the output of the algorithm.

\begin{lemma} \label{lem:generalMatroid_s_larger}
If Algorithm~\ref{alg:general_aided} selects a value $p$ and $|S| \geq |OPT_{p + 3}| / [2(d + 1)]$, then $f(S) \geq w(OPT_{p + 3}) / [32(d + 1)^2]$.
\end{lemma}

Intuitively, the last lemma holds since a large $|S|$ means that the algorithm adds elements to $S$ in many iterations, and each iteration increases $f(S)$ by at least $\tau$.

\begin{lemma} \label{lem:generalMatroid_s_smaller}
If Algorithm~\ref{alg:general_aided} selects a value $p$ and $|S| < |OPT_{p + 3}| / [2(d + 1)]$, then $f(S) \geq w(OPT_{p + 3}) / 8$.
\end{lemma}

The main idea behind the proof of Lemma~\ref{lem:generalMatroid_s_smaller} is as follows. Since $|S|$ is small, many elements of $OPT_{p + 3} \setminus S$ could be added to it, together with their dependencies, without violating independence. The reason these elements were not added must be that they did not pass the threshold, which can only happen when $f(S)$ is large enough.

\begin{corollary} \label{cor:value_if_p}
If Algorithm~\ref{alg:general_aided} selects a value $p$, then $f(S) \geq w(OPT_{p + 3}) / [32(d + 1)^2]$.
\end{corollary}

We are now ready to prove Theorem~\ref{thm:general_aided}.

\begin{proof}[Proof of Theorem~\ref{thm:general_aided}]
Recall that every value $p$ is selected by Algorithm~\ref{alg:general_aided} with probability at least $(\log_2 \alpha + \log_2 k + 6)^{-1} = (\log_2 (\alpha k) + 6)^{-1}$. Hence, by Corollary~\ref{cor:value_if_p}, the expected value of the output of Algorithm~\ref{alg:general_aided} is at least:
\[
	\frac{1}{\log_2(\alpha k) + 6} \cdot \sum_{p = p_1}^{p_2} \frac{w(OPT_{p + 3})}{32(d + 1)^2}
	=
	\frac{w\left(\bigcup_{p = p_1}^{p_2} OPT_p\right)}{32(d + 1)^2 \cdot [\log_2(\alpha k) + 6]}
	\geq
	\frac{f(OPT)}{64(d + 1)^2 \cdot [\log_2(\alpha k) + 6]}
	\enspace,
\]
where the last inequality is due to Lemma~\ref{lem:sum_weights}.
\end{proof}

\section{Estimating the optimum: from aided to non-aided algorithms} \label{sec:estimation}
In this section, we show how to convert aided algorithms into non-aided ones. Together with our aided algorithms, the following theorem implies the results stated in Theorems~\ref{t:general_nonpoly} and~\ref{t:uniform_nonpoly}.
\begin{theorem} \label{thm:multiple_elements_approximation}
If there exists a $(80(d + 2)^2)$-aided $\beta$-competitive algorithm $ALG$ for the monotone matroid secretary problem with supermodular degree $d$, under a class $\cC$ of matroid constraints closed under restriction, then there also exists a non-aided $O(d^3\log d + \beta)$-competitive algorithm for the same problem.
\end{theorem}

Recall that the truncation of a matroid $M = (\cN, \cI)$ to rank $k'$ is a matroid $M' = (\cN, \cI')$ where a set $S \subseteq \cN$ is independent in $M'$ if and only if $S \in \cI$ and $|S| \leq k'$.
The algorithm we use to prove Theorem~\ref{thm:multiple_elements_approximation} is Algorithm~\ref{alg:multiple_elements_approximation}.

\SetKwIF{With}{OtherwiseWith}{Otherwise}{with}{do}{otherwise with}{otherwise}{}
\SetKwFor{InfiniteRepeat}{repeat}{}{end}
\begin{algorithm}[H]
\DontPrintSemicolon
\caption{\textsf{Multiple Elements Estimation}} \label{alg:multiple_elements_approximation}
\With{probability $1/2$}
{
	Apply the algorithm guaranteed by Theorem~\ref{thm:smallRankMatroid} to the problem after truncating the matroid to rank $\min\{k, d + 1\}$ (where $k$ is the rank of the original matroid).\\
}
\Otherwise
{
	Choose $X$ according to the binomial distribution $B(n, (d + 2)^{-1})$, and let $T$ the set of the first $X$ elements revealed.\\
	Let $A \gets \varnothing$ and $W \gets 0$.\\
	\While{there exist $u \in T \setminus A$ and $\cD_u \subseteq \DS{u}$ s.t. $A \cup \cD_u + u \in \cI$}
	{ \label{line:loop}
		Find such a pair maximizing $f(u \mid A \cup \cD_u)$.\\	
			Increase $W \gets W + f(u \mid A \cup \cD_u)$.\\
			Update $A \gets A \cup \cD_u + u$.
	}
	Apply $ALG$ to the remaining elements with $\opt_{80(d + 2)^2} = W / 10$.\\
}
\end{algorithm}

Algorithm~\ref{alg:multiple_elements_approximation} consists of two parts, each executed with probability $1/2$. In order to prove Theorem~\ref{thm:multiple_elements_approximation} we show that for any instance of the monotone matroid secretary problem, one of the following cases is true:
\begin{itemize} \setlength{\itemsep}{0mm}
\item The algorithm guaranteed by Theorem~\ref{thm:smallRankMatroid} produces an $O(d^3 \log d)$-competitive solution for the non-truncated problem.
\item The second part of Algorithm~\ref{alg:multiple_elements_approximation} is $O(\beta)$-competitive.
\end{itemize}

To determine which of the above cases is true for every given instance we need some notation. Let $u^* \in \cN$ be an element maximizing
\[	\max_{\substack{S \subseteq \DS{u^*} \\ S + u^* \in \cI}} f(u^* \mid S)
	\enspace,
\]
and let $m^*$ and $S^*$ denote the value of this maximum and an arbitrary corresponding set $S$, respectively. It can be shown quite easily that the first case above holds when $m^* \geq f(OPT)/(256(d + 1)^2)$ (see Lemma~\ref{lem:big_multiple} for details). Thus, we sketch here only the more interesting part of the analysis, which is to show that the second above case holds when $m^* \leq f(OPT)/(256(d + 1)^2)$ (a full proof can be found in Appendix~\ref{app:FullProofSmallCaseNewEstimation}).
The main thing that we need to show is that $\opt_{80(d + 2)^2}$ is, with constant probability, a $80(d + 1)^2$-estimation for $f(OPT)$. For that purpose, we relate the expected value of the estimate $\opt_{80(d + 2)^2}$ to $f(OPT)$, and then bound the variance of this estimate to show that it is close enough to its expected value with constant probability.

In the rest of this section we assume Algorithm~\ref{alg:multiple_elements_approximation} executes its second part. Let $W_{\ell}$ and $A_{\ell}$ be $W$ and $A$, respectively, when Algorithm~\ref{alg:multiple_elements_approximation} exits its loop.
We bound $W_{\ell}$, which immediately implies bounds for $\opt_{80(d + 2)^2}$.
In order to achieve that, we switch our attention from Algorithm~\ref{alg:multiple_elements_approximation} to an offline algorithm (Algorithm~\ref{alg:offline_W}) producing exactly the same distribution for $W_{\ell}$.
We explain intuitively why the distributions of $W_{\ell}$ in both algorithms are the same.
Let us choose a set $T_{\text{pre}}$ ahead, exactly as $T$ is chosen by the online algorithm. 
Next, we modify the offline algorithm so that whenever it chooses an element from its set $T$, instead of randomly deciding whether to keep it in $T$, it queries membership in $T_{\text{pre}}$.
Clearly, since there are no repetitions in these queries, this does not affect the behavior of the offline algorithm. However, one can verify that the output of the offline algorithm is now identical to the output of the online algorithm when it selects $T = T_{\text{pre}}$. 

\begin{algorithm}[ht]
\DontPrintSemicolon
\caption{\textsf{Offline W Calculation}} \label{alg:offline_W}
Let $A \gets \varnothing$, $W \gets 0$ and $T \gets \cN$.\\
\While{there exist $u \in T \setminus A$ and $\cD_u \subseteq \DS{u}$ s.t. $A \cup \cD_u + u \in \cI$}
{
	Find such a pair maximizing $f(u \mid A \cup \cD_u)$.\\	
	
	\With{probability $(d + 2)^{-1}$}
	{
		Increase $W \gets W + f(u \mid A \cup \cD_u)$.\\
		Update $A \gets A \cup \cD_u + u$.
	}
	\lOtherwise{Update $T \gets T - u$.}
	
}
\end{algorithm}

Let $L_{\ell}$ be the sum of $f(u \mid A \cup \cD_u)$ for all the iterations done by Algorithm~\ref{alg:offline_W}, regardless of the random choice made by the algorithm.
We show how to lower bound the expectation of $L_{\ell}$ with respect to $f(OPT)$, and then use a bound on the variance of $W_{\ell}$ to get a concentration result for $W_{\ell}$.

\begin{lemma}
$(d + 1) \cdot L_{\ell} \geq f(A_{\ell})$.
\end{lemma}
\begin{proof} [Brief sketch of proof]
The proof is by induction on the number of iterations. Assume the lemma is true for $i-1$ iterations, and let us prove it for iteration $i$.
Trivially, if Algorithm~\ref{alg:offline_W} randomly chooses not to add elements to $A$, then the lemma is true for iteration $i$ as well.
Now, assume the random choice made is to add the elements to $A$.
Note that only up to $d+1$ elements are added to $A$, and therefore, it is sufficient to show that the marginal value of each is upper bounded by the increase in the value of $L$.
Let $u$ be the element chosen by the algorithm. Any dependency of $u$ that appears in $T$ could also be chosen instead of $u$, so its marginal is upper bounded by the increase in the value of $L$ (since the algorithm uses a greedy choice).
On the other hand, any dependency $u' \not \in T$ must have been removed when chosen by the algorithm in a previous iteration. Therefore, its marginal value, computed with respect to its optimal dependencies in this previous iteration, is already counted by $L_{\ell}$. Note that the last marginal value must be at least as large as the marginal value of $u'$ with respect to the optimal dependencies in the current iteration (by definition of supermodular dependencies).
\end{proof}

\begin{lemma}
$f(A_{\ell}) + (d + 1) \cdot L_{\ell} \geq f(OPT)$.
\end{lemma}
\begin{proof} [Brief sketch of proof]
We consider a hybrid solution starting as $OPT$ and ending as $A_{\ell}$. We use the matroid augmentation property to observe that, when a new element of $A$ is added to this hybrid solution, no more than $d+1$ non $A$ elements have to be removed to restore independence. Then, we bound the ``damage'' resulting from the removal of these elements using the greedy choice of the algorithm and the definition of supermodular dependencies.
Finally, we observe that the value of $A_{\ell}$ cannot be increased by adding elements that are still in $T$, since, otherwise, the algorithm would have done that.
This means that $A_{\ell}$ itself is as good as the final hybrid (which contains it).
\end{proof}
An immediate corollary of the last two lemmata is a lower bound of $f(OPT) / (2(d + 1))$ on $L_{\ell}$.
This bound, together with a concentration result we prove in Appendix~\ref{app:FullProofSmallCaseNewEstimation}, shows that $W_{\ell} \geq  f(OPT) / O(d^2)$ with constant probability.
The last inequality implies with constant probability, when $m^* \leq f(OPT)/(256(d + 1)^2)$, that $\opt_{80(d + 2)^2} = W / 10$ is indeed a $(80(d + 2)^2)$-estimation for the optimum of the part of the input that was not read by the estimation algorithm (\ie, the input for the aided algorithm). The competitive ratio of the second part of Algorithm~\ref{alg:multiple_elements_approximation} then follows from the competitive ratio of the aided algorithm.

\paragraph{Acknowledgment.}
We are grateful to Uri Feige for valuable discussions.

\apptocmd{\sloppy}{\hbadness 10000\relax}{}{}
\bibliographystyle{plain}
\bibliography{supermodular}

\appendix

\vspace{-0.1in}
\section{Assuming our set functions are normalized is without loss of generality} \label{app:normalizedReduction}
\vspace{-0.1in}

\begin{reduction} \label{r:normalized}
If $ALG$ is an $\alpha$-competitive algorithm for the monotone matroid secretary problem under the assumption that $f$ is normalized, then $ALG$ is an $\alpha$-competitive algorithm also without this assumption.
\end{reduction}
\begin{proof}
Let $g\colon 2^\cN \to \nonnegR$ be the function $g(S) = f(S) - f(\varnothing)$. Notice that $g$ is a non-negative monotone function and $\DSF{f} = \DSF{g}$. Moreover, all the oracles that an algorithm for the monotone matroid secretary problem has access to return the same answers for both $f$ and $g$, and thus, the algorithm produces a random set $S$ with the same distribution when given either $f$ or $g$ as input. Since $g$ is normalized, by the definition of $ALG$:
\[
	\bE[g(S)]
	\geq
	\frac{g(OPT)}{\alpha}
	\enspace,
\]
where $OPT$ is a set maximizing $g$ (and $f$). Thus:
\[
	\bE[f(S)]
	=
	\bE[g(S)] + f(\varnothing)
	\geq
	\frac{g(OPT)}{\alpha} + f(\varnothing)
	\geq
	\frac{f(OPT)}{\alpha}
	\enspace.
	\qedhere
\]
\end{proof}

\vspace{-0.1in}
\section{Small rank matroids (Proof of Theorem~\ref{thm:smallRankMatroid})} \label{app:small_rank}
\vspace{-0.1in}

In this proof we need some additional notation. The max-marginal of an element $u$ at time $i$ is the largest marginal value that $u$ can contribute to a subset of the elements that arrive after time $i$ (while keeping the subset independent). More formally, let $\cN_i$ be the (random) set of the first $i$ elements that arrived, then the max-marginal of an element $u$ at time $i$ is:
\[
	\maxm(u, i) = \max_{\substack{S \subseteq \cN \setminus \cN_i \\ S + u \in \cI}} f(u \mid S)
	\enspace.
\]
We also use $\maxs(u, i)$ to denote an arbitrary set for which the maximum is obtained. Note that one can calculate both $\maxm(u, i)$ and $\maxs(u, i)$ in $O(2^d)$ time. We begin the proof with the following simple claim.

\begin{claim} \label{c:n_good}
We may assume that $n$ is dividable by any quantity $h$ whose value is polynomial in $k$.
\end{claim}
\begin{proof} 
Let $n'$ be the least multiple of $h$ which is at least as large as $n$. Note that $n'$ is polynomial in $n$ (since $k \leq n$). Let $\cN'$ be a ground set containing the elements of $\cN$ and a set $D$ of $n' - n$ dummy elements. We extend $f$ and $M$ to $\cN'$ as follows:
\begin{itemize}
	\item The function $f'\colon 2^{\cN'} \to \bR^+$ is defined as: $f'(S) = f(S \setminus D)$ for every set $S \subseteq \cN'$. Note that $f'$ is non-negative, monotone and has a supermodular degree of $d$. Additionally, $\DSS{f'}{u} = \varnothing$ for the dummy elements of $D$, and $\DSS{f'}{u} = \DSS{f}{u}$ for every other element.
	\item The matroid $M' = (\cN', \cI')$ is defined by the following rule. A set $S \subseteq \cN'$ is in $\cI'$ if and only if $S \setminus D \in \cI$ and $|S| \leq k$. Note that this rule defines a matroid of rank $k$ which is uniform whenever $M$ is.
\end{itemize}

One can observe that the problems $(f, M)$ and $(f', M')$ are equivalent in the sense that: any solution for $(f, M)$ is also a solution for $(f', M')$ of the same value, and removing the dummy elements of any solution for $(f', M')$ results in a solution for $(f, M)$ of the same value. Moreover, given access to the oracles corresponding to $(f, M)$, one can efficiently implement the oracles for $(f', M')$. Thus, given algorithm $ALG$ that is $r$-competitive for ground sets obeying the requirements of the reduction, one can construct an $r$-competitive algorithm for general ground sets as follows:
\begin{compactenum}
	\item Apply $ALG$ to the instance $(f', M')$.
	\item Accept every element of $\cN = \cN' \setminus D$ that $ALG$ accepts.
	\qedhere
\end{compactenum}
\end{proof}

In the rest of this section we make two assumptions. First, we assume that $n$ is dividable by $10k$, which is justified by Claim~\ref{c:n_good}. Second, we assume $k \geq 2$ (if $k = 1$, then the classical secretary algorithm can be used to get an $O(1)$-competitive algorithm). Our objective is to show that Algorithm~\ref{alg:random_threshold} obeys the requirements of Theorem~\ref{thm:smallRankMatroid} given these assumptions.

\begin{algorithm}[H]
\caption{\textsf{Small Rank Matroid}} \label{alg:random_threshold}
Select an arbitrary order $\prec$ over the elements of the ground set $\cN$.\\
Let $p$ be a uniformly random integer from the set $\{0, 1, \dotsc, \lceil \log_2 k \rceil\}$.\\
Reject the first $t = 2^p \cdot \frac{n}{2k}$ elements.\\
\For{$i$ = $t + 1$ to $n$}
{
	Let $u_i$ be the element arriving at time $i$.\\
	\If{for every element $u \in \cN_{i - 1}$ either $\maxm(u_i, i) > \maxm(u, i)$ or\\ \hspace{2.5in} ($\maxm(u_i, i) = \maxm(u, i)$ and $u_i \succ u$) \label{lin:condition}}
	{
		Terminate the ``for'' loop and accept the elements of $\maxs(u_i, i) + u_i$ when they arrive.\\
	}
}
\end{algorithm}

For the purpose of analyzing Algorithm~\ref{alg:random_threshold}, it is helpful to think about the input as created backwards by the following process. The set $\cN_n$ is simply the entire ground set $\cN$. Then, the last element of the input $u_n$ is selected uniformly at random from $\cN_n$, and the set $\cN_{n - 1}$ becomes $\cN_n - u_n$. On the next step, the $(n-1)$-th element $u_{n - 1}$ of the input is selected uniformly at random from $\cN_{n - 1}$ and we set $\cN_{n - 2} = \cN_{n - 1} - u_{n - 1}$. The process than continuous in the same way, \ie, when it is time to determine the $i$-th element of the input, this element is selected uniformly at random from $\cN_i$, and we set $\cN_{i - 1} = \cN_i - u_i$.

We say that an element $u$ is the \textsf{top} element of a set $\cN_i$ if for every other element $u' \in \cN_i \setminus u$ either $\maxm(u', i) < \maxm(u, i)$ or $\maxm(u', i) = \maxm(u, i)$ and $u' \prec u$. Note that Line~\ref{lin:condition} of Algorithm~\ref{alg:random_threshold} in fact checks whether $u_i$ is the top element of $\cN_i$. Additionally, we say that an input is \emph{well-behaved} with respect to the value $t$ and order $\prec$ chosen by Algorithm~\ref{alg:random_threshold} if it has the following properties:
\begin{compactenum}[({A}1)]
	\item There exists a time $i > n / k$ such that for some element $u \in \cN_i$, $\maxm(u, i) \geq f(OPT) / k$, where $OPT$ is an independent set maximizing $f$. We denote the first such time by $\ell_1$. \label{prop:there_is_good}
	\item There exists exactly a single time $t < i \leq \ell_1$ such that $u_i$ is the top element of $\cN_i$.\footnote{Note that in some cases we might have $\ell_1 \leq t$. In these cases the input is not well-behaved with respect to $t$ and~$\prec$.} We denote this time by $\ell_2$. \label{prop:single_top}
\end{compactenum}

The analysis of Algorithm~\ref{alg:random_threshold} consists of two parts. First we show that it produces a good output for well-behaved inputs, and then we show that the input is well-behaved with a significant probability.

\begin{lemma} \label{lem:when_good}
Algorithm~\ref{alg:random_threshold} outputs a solution of value at least $f(OPT)/k$ when its input is well-behaved with respect to $t$ and $\prec$.
\end{lemma}
\begin{proof}
The definition of the algorithm and Property~A\ref{prop:single_top} guarantees that the algorithm outputs $\maxs(u_{\ell_2}, \ell_2) + u_{\ell_2}$. The value of this solution is:
\[
	f(\maxs(u_{\ell_2}, \ell_2) + u_{\ell_2})
	\geq
	f(u_{\ell_2} \mid \maxs(u_{\ell_2}, \ell_2))
	=
	\maxm(u_{\ell_2}, \ell_2)
	\enspace.
\]
Hence, we are only left to lower bound $\maxm(u_{\ell_2}, \ell_2)$. Observe that by Property~A\ref{prop:there_is_good}, there exists an element $u'_{\ell_1} \in \cN_{\ell_1}$ such that $\maxm(u'_{\ell_1}, \ell_1) \geq f(OPT)/k$. Let us prove by a backward induction that this is true for every $\ell_2 \leq i \leq \ell_1$, \ie, that for every such time there exists an element $u'_{i} \in \cN_i$ such that $\maxm(u'_i, i) \geq f(OPT)/k$.

Assume the claim holds for a given $\ell_2 < i \leq \ell_1$, and let us prove it for $i - 1$. Observe that we can assume without loss of generality that $u'_i$ is the top element of $\cN_i$. Thus, by Property~A\ref{prop:single_top} $u'_i \neq u_i$, which implies: $u'_i \in \cN_{i - 1}$. By definition $\maxm(u, i)$ is a non-increasing function of $i$, hence, $\maxm(u'_i, i - 1) \geq \maxm(u'_i, i) \geq f(OPT)/k$, which complete the induction step.

The claim we proved by induction implies: $\maxm(u'_{\ell_2}, \ell_2) \geq f(OPT)/k$. The lemma now follows by observing that Property~A\ref{prop:single_top} guarantees that $u_{\ell_2}$ is the top element of $\cN_{\ell_2}$, and thus, $\maxm(u_{\ell_2}, \ell_2) \geq \maxm(u'_{\ell_2}, \ell_2)$.
\end{proof}

\begin{lemma} \label{lem:there_is_good_prob}
Property~A\ref{prop:there_is_good} holds with a probability of at least $0.2$.
\end{lemma}

\begin{proof}
Given an element $u \in \cN$, let $i_u$ denote the time when it arrives. Then,
\[
	\sum_{u \in OPT} \maxm(u, i_u)
	\geq
	\sum_{u \in OPT} f(u \mid OPT \setminus \cN_u)
	=
	f(OPT)
	\enspace,
\]
where the inequality follows from the definition of $\maxm$. Since $|OPT| \leq k$, we get by averaging that for some element $u \in OPT$ there must be $\maxm(u, i_u) \geq f(OPT)/k$. Hence, Property~A\ref{prop:there_is_good} is guaranteed to hold when all the elements of $OPT$ appear after time $\hat{t} = n/k$. The last event occurs with a probability of at least:
\begin{align*}
	\frac{\left(\hat{t}! \cdot \binom{n - k}{\hat{t}}\right) \cdot (n - \hat{t})!}{n!}
	={} &
	\frac{(n - k)! \cdot (n - \hat{t})!}{n! \cdot(n - k - \hat{t})!}
	=
	\prod_{i = 0}^{\hat{t} - 1} \frac{n - k - i}{n - i}\\
	\geq{} &
	\left(\frac{n - k - \hat{t}}{n - \hat{t}}\right)^{\hat{t}}
	=
	\left(1 - \frac{k}{n - \hat{t}}\right)^{\hat{t}}
	\geq
	\left(1 - \frac{k}{0.9n}\right)^{n/k}\\
	\geq{} &
	e^{-1/0.9} \cdot \left(1 - \frac{k}{0.9^2 \cdot n}\right)
	\geq
	e^{-10/9} \cdot \left(1 - \frac{1}{8.1}\right)
	\geq
	0.2
	\enspace.
	\qedhere
\end{align*}
\end{proof}

\begin{lemma} \label{lem:single_top_prob_random}
Given that Property~A\ref{prop:there_is_good} holds, Property~A\ref{prop:single_top} holds with a probability of at least $(\log_2 k + 2)^{-1}/4$.
\end{lemma}
\begin{proof}
First, let us consider the event $E_1$ that there exists a time $\ell_1/2 < \ell'_2 \leq \ell_1$ such that $u_{\ell'_2}$ is the top element of $\cN_{\ell'_2}$ and for every time $\ell_2' < i \leq \ell_1$, $u_i$ is not the top element of $\cN_i$. For this event not to occur, a non-top element $u_i$ must be selected from $\cN_i$ for every time $\ell_1/2 < i \leq \ell_1$, which happens with probability:
\[
	\prod_{i = \lfloor \ell_1/2 + 1 \rfloor}^{\ell_1} \frac{i - 1}{i}
	=
	\frac{\lfloor \ell_1/2 \rfloor}{\ell_1}
	\leq
	\frac{1}{2}
	\enspace.
\]
Hence, $E_1$ occurs with the complement probability, which is at least $1/2$. Next, given that $E_1$ occurred, we are interested in the event that $\ell'_2/2 \leq t < \ell'_2$, which we denote by $E_2$. It is important to notice that $E_1$ is independent of the choice of $t$ by the algorithm, and thus, the distribution of $t$ is unaffected by conditioning on $E_1$. Additionally, notice that:
\[
	2^0 \cdot \frac{n}{2k}
	=
	\frac{n}{2k}
	\leq
	\frac{\ell_1}{2}
	<
	\ell'_2
	\qquad
	\text{and}
	\qquad
	\frac{\ell'_2}{2}
	\leq
	\frac{n}{2}
	\leq
	2^{\lceil \log_2 k \rceil} \cdot \frac{n}{2k}
	\enspace.
\]
Hence, one of the possible values of $t$ obeys the requirement $\ell'_2/2 \leq t < \ell'_2$. Since $t$ takes at most $\log_2 k + 2$ different values, and it takes them with equal probabilities, we get that $E_2$ occurs with a probability of at least $(\log_2 k + 2)^{-1}$ given $E_1$.

Given that $E_1$ and $E_2$ both occur, for Property~A\ref{prop:single_top} to hold with need the additional event that in the range $(t, \ell'_2)$ no element $u_i$ is the top element of $\cN_i$. Note that the order of the elements of $\cN_{\ell'_2} - u_{\ell'_2}$ is independent of $E_1$ and $E_2$. Hence, the probability of this event is at least:
\[
	\prod_{i = t + 1}^{\ell'_2-1} \frac{i - 1}{i}
	=
	\frac{t}{\ell'_2 - 1}
	\geq
	\frac{\ell'_2/2}{\ell'_2 - 1}
	>
	1/2
	\enspace.
	\qedhere
\]
\end{proof}

We are now ready to prove Theorem~\ref{thm:smallRankMatroid}.

\begin{proof}[Proof of Theorem~\ref{thm:smallRankMatroid}]
Lemmata~\ref{lem:there_is_good_prob} and~\ref{lem:single_top_prob_random} imply that the input is well-behaved with respect to $t$ and $\prec$ with probability at least $(\log_2 k + 2)^{-1}/20$. By Lemma~\ref{lem:when_good}, when the input is well-behaved with respect to $t$ and $\prec$, Algorithm~\ref{alg:random_threshold} outputs a solution of value at least $f(OPT) / k$. Hence, the competitive ratio of Algorithm~\ref{alg:random_threshold} is at least:
\[
	20k(\log_2 k + 2)
	=
	O(k \log k)
	\enspace.
	\qedhere
\]
\end{proof}

\section{Missing proofs} \label{app:missing_proofs}

This section contains proofs that have been omitted from the main body of the paper.

\begin{proof}[Proof of Lemma~\ref{lem:sum_weights}]
Clearly $w(OPT) = f(OPT)$. Moreover, by definition, $2^{p_1} \cdot \opt_\alpha \leq f(OPT)/k$.
Hence:
\[
	f(OPT) - w\left(\bigcup_{p = p_1}^{p_2} OPT_p\right)
	=
	\sum_{\substack{w \in OPT \\ w(u) < 2^{p_1} \cdot \opt_\alpha/2}} \mspace{-18mu} w(u)
	\leq
	k \cdot (2^{p_1} \cdot \opt_\alpha/2)
	\leq
	\frac{f(OPT)}{2}
	\enspace.
	\qedhere
\]
\end{proof}

\begin{proof}[Proof of Lemma~\ref{lem:generalMatroid_s_larger}]
For an element $u \in \cN$, let $S_u$ be the set $S$ immediately before $u$ is processed by Algorithm~\ref{alg:general_aided}.  Note that each time that Algorithm~\ref{alg:general_aided} adds elements to $S$, it adds up to $d + 1$ elements and $f(S)$ increases by at least $\tau$ since, by monotonicity:
\[
	f(\DO{u} + u \mid S_u)
	\geq
	f(u \mid \DO{u} \cup S_u)
	\enspace.
\]
Hence, we can lower bound $f(S)$ by:
\[
	f(S)
	\geq
	\left\lceil \frac{|S|}{d + 1} \right\rceil \cdot \tau
	\geq
	\frac{|OPT_{p + 3}|}{2(d + 1)^2} \cdot \left[\frac{1}{16} \cdot \max_{u \in OPT_{p + 3}} w(u)\right]
	\geq
	\frac{w(OPT_{p + 3})}{32(d + 1)^2}
	\enspace.
	\qedhere
\]
\end{proof}

\begin{proof}[Proof of Lemma~\ref{lem:generalMatroid_s_smaller}]
Observe that $OPT_{p + 3}$ is a subset of $OPT$, and thus, independent. Hence, by the matroid properties, there exists a set $O' \subseteq OPT_{p + 3} \setminus S$ of size at least $|OPT_{p + 3}| - |S|$ such that $O' \cup S \in \cI$. Every element $u \in OPT_{p + 3} \setminus O'$ can belong to the dependence set of at most $d$ other elements of $OPT_{p + 3}$. Thus, the number of elements $u \in OPT_{p + 3}$ having $\DS{u} \cap OPT + u \not \subseteq O'$ is upper bounded by:
\[
	(d + 1) \cdot |S|
	<
	\frac{|OPT_{p + 3}|}{2}
	\enspace.
\]
In other words, there exists a set $O'' \subseteq OPT_{p + 3}$ of size at least $|OPT_{p + 3}| / 2$ such that $\DS{u} \cap OPT + u \subseteq O'$ for every $u \in O''$. Observe that by monotonicity:
\begin{align*}
	f(O')
	\geq{} &
	\sum_{u \in O''} f(u \mid O' \setminus \cN_u)
	\geq
	\sum_{u \in O''} f(u \mid OPT \setminus \cN_u)\\
	={} &
	w(O'')
	\geq
	\frac{|OPT_{p + 3}|}{2} \cdot \min_{u \in OPT_{p + 3}} w(u)
	\geq
	\frac{w(OPT_{p + 3})}{4}
	\enspace,
\end{align*}
where the second inequality holds since $O'$ already contains all the elements of $\DS{u} \cap OPT$.

Every element $u \in O'$ must have been rejected upon arrival by Algorithm~\ref{alg:general_aided} due to the threshold. Moreover, for every such element $u$ we have $([\DS{u} \cap (O' \cup S)] \setminus \cN_u + u) \cup S_u \subseteq O' \cup S \in \cI$ (where $S_u$ is, again, the set $S$ immediately before $u$ is processed by Algorithm~\ref{alg:general_aided}). Hence:
\[
	f(u \mid (O' \setminus \cN_u) \cup S)
	\leq
	f(u \mid [\DS{u} \cap (O' \cup S)] \setminus \cN_u \cup S_u)
	<
	\tau
	\enspace,
\]
where the first inequality holds by the definition of $\DS{u}$. Adding the last inequality over all elements $u \in O'$ gives:
\begin{align*}
	\frac{w(OPT_{p + 3})}{4}
	\leq{} &
	f(O')
	\leq
	f(O' \cup S)
	=
	f(S) + \sum_{u \in O'} f(u \mid (O' \setminus \cN_u) \cup S)
	<
	f(S) + |OPT_{p + 3}| \cdot \tau\\
	\leq{} &
	f(S) + |OPT_{p + 3}| \cdot \frac{\min_{u \in OPT_{p + 3}} w(u)}{8}
	\leq
	f(S) + \frac{w(OPT_{p + 3})}{8}
	\enspace.
\end{align*}
The lemma now follows by rearranging the last inequality.
\end{proof}

\begin{lemma} \label{lem:big_multiple}
If $m^* \geq f(OPT)/[256(d + 1)^2]$, then Algorithm~\ref{alg:multiple_elements_approximation} is $O(d^3 \log d)$-competitive.
\end{lemma}
\begin{proof}
Note that $S^* + u^*$ is an independent set in the matroid even after it is truncated to rank $\min\{k, d + 1\}$. Hence, when Algorithm~\ref{alg:multiple_elements_approximation} applies the algorithm guaranteed by Theorem~\ref{thm:smallRankMatroid} (which happens with probability $1/2$), the produced set has an expected value of at least:
\[
	\frac{f(S^* + u^*)}{O(d \log d)}
	\geq
	\frac{f(u^* \mid S^*)}{O(d \log d)}
	=
	\frac{m^*}{O(d \log d)}
	\geq
	\frac{f(OPT) / [256(d + 1)^2]}{O(d \log d)}
	\enspace.
	\qedhere
\]
\end{proof}

\vspace{-0.1in}
\section{Estimation aided algorithm for a uniform matroid constraint} \label{app:Aided_UniformMatroid}
\vspace{-0.1in}

In this section we prove the following theorem.

\begin{theorem} \label{thm:aided_uniform_general_alpha}
For every $\alpha \geq 1$, there exists an $\alpha$-aided $O(d \log \alpha)$-competitive algorithm for the monotone matroid secretary problem when the matroid $M$ is uniform.
\end{theorem}

Before proving the existence of an $\alpha$-aided algorithms for any $\alpha \geq 1$, let us begin with a $2$-aided algorithm.

\begin{proposition} \label{p:uniform_aided_2}
There exists a $2$-aided $O(d)$-competitive algorithm for the monotone matroid secretary problem when the matroid $M$ is uniform.
\end{proposition}

The algorithm we use to prove Proposition~\ref{p:uniform_aided_2} is Algorithm~\ref{alg:cardinality_aided_2}.

\begin{algorithm}[h!t]
\caption{\textsf{$2$-Aided Cardinality}} \label{alg:cardinality_aided_2}
Let $\tau \gets \frac{\opt_2}{2k}$.\\
Let $S \gets \varnothing$.\\
\For{every arriving element $u$}
{
	\If{there exits a set $\DO{u} \subseteq \DS{u} \setminus \cN_u$ such that $f(u \mid \DO{u} \cup S) \geq \tau$ and\\
	\hspace{2.5in} $|S| + |\DO{u}| + 1 \leq k$}
	{
		Add $\DO{u} + u$ to $S$.
	}
}
\Return{$S$}.
\end{algorithm}

Let $S_u$ be the set $S$ before the element $u$ is processed. When $S$ appears below without a subscript it denotes the output of the algorithm.

\begin{lemma}
If $|S| \leq \max\{0, k - d - 1\}$, then $f(S) \geq f(OPT)/2$.
\end{lemma}
\begin{proof}
Assume, towards a contradiction, that $|S| \leq \max\{0, k - d - 1\}$ and still $f(S) < f(OPT)/2$. Since $|S| \leq \max\{0, k - d - 1\}$, for every element $u \in OPT \setminus S$ Algorithm~\ref{alg:cardinality_aided_2} could add the set $[\DS{u} \cap (OPT \cup S)] \setminus \cN_u + u$ to $S$. From the fact that the algorithm did not add this set (or any other set containing $u$) to $S$, we learn that:
\[
	f(u \mid (OPT \setminus \cN_u) \cup S)
	\leq
	f(u \mid [\DS{u} \cap (OPT \cup S)] \setminus \cN_u \cup S_u)
	<
	\tau
	\enspace,
\]
where the first inequality holds by the definition of $\DS{u}$. Adding the last inequality over all elements $u \in OPT \setminus S$ gives:
\[
	f(OPT)
	\leq
	f(OPT \cup S)
	=
	f(S) + \sum_{u \in OPT} f(u \mid (OPT \setminus \cN_u) \cup S)
	<
	f(S) + k \tau
	\enspace.
\]
Plugging the assumption that $f(S) < f(OPT)/2$ and the definition of $\tau$ into the last inequality gives an immediate contradiction.
\end{proof}

\begin{lemma}
If $|S| \geq \max\{1, k - d\}$, then $f(S) \geq f(OPT) / [8(d + 1)]$.
\end{lemma}
\begin{proof}
Note that each time that Algorithm~\ref{alg:cardinality_aided_2} adds elements to $S$, it adds up to $d + 1$ elements and $f(S)$ increases by at least $\tau$ since, by monotonicity:
\[
	f(\DO{u} + u \mid S_u)
	\geq
	f(u \mid \DO{u} \cup S_u)
	\enspace.
\]
Hence, we can lower bound $f(S)$ by:
\[
	f(S)
	\geq
	\left\lceil \frac{|S|}{d + 1} \right\rceil \cdot \tau
	\geq
	\left\lceil \frac{\max\{1, k - d\}}{d + 1} \right\rceil \cdot \frac{\opt_2}{2k}
	\geq
	\frac{k}{2(d + 1)} \cdot \frac{f(OPT)}{4k}
	=
	\frac{f(OPT)}{8(d + 1)}
	\enspace.
	\qedhere
\]
\end{proof}

Proposition~\ref{p:uniform_aided_2} follows immediately from the last two lemmata. Theorem~\ref{thm:aided_uniform_general_alpha} generalizes Proposition~\ref{p:uniform_aided_2} to general $\alpha$-aided algorithms. The algorithm we use to prove Theorem~\ref{thm:aided_uniform_general_alpha} is Algorithm~\ref{alg:alpha_to_2}.

\begin{algorithm}[h!t]
\caption{\textsf{$\alpha$-Aided Cardinality}} \label{alg:alpha_to_2}
Let $p$ be a uniformly random integer from the set $\{0, 1, \dotsc, \lceil \log_2 \alpha \rceil\}$.\\
Apply Algorithm~\ref{alg:cardinality_aided_2} with $\opt_2 = 2^p \cdot \opt_\alpha$.\\
\end{algorithm}

\begin{proof}[Proof of Theorem~\ref{thm:aided_uniform_general_alpha}]
The largest value Algorithm~\ref{alg:alpha_to_2} can assign to $\opt_2$ is:
\[
	2^{\lceil \log \alpha \rceil} \cdot \opt_\alpha
	\geq
	\alpha \cdot (f(OPT)/\alpha)
	=
	f(OPT)
	\enspace.
\]
On the other hand, the smallest value Algorithm~\ref{alg:alpha_to_2} can assign to $\opt_2$ is: $2^0 \cdot \opt_\alpha \leq	f(OPT)$. Hence, for some $p$ Algorithm~\ref{alg:alpha_to_2} is guaranteed to produce a value $\opt_2$ obeying $f(OPT)/2 \leq \opt_2 \leq f(OPT)$, in which case Algorithm~\ref{alg:cardinality_aided_2} is $O(d)$-competitive by Proposition~\ref{p:uniform_aided_2}. Since every value of $p$ occurs with a probability of at least $1 / (\log_2 \alpha + 2)$, we get that the competitive ratio of Algorithm~\ref{alg:alpha_to_2} is at most $O(d \log \alpha)$.
\end{proof}

\section{Full proof for {\boldmath $m^* \leq f(OPT)/(256(d + 1)^2)$}} \label{app:FullProofSmallCaseNewEstimation}
In this section we analyze Algorithm~\ref{alg:multiple_elements_approximation} in the case of a small $m^*$. We begin with a concentration result proved in Section~\ref{sec:process}. The analysis of Algorithm~\ref{alg:multiple_elements_approximation} appears in Section~\ref{ssc:proof_small_m}.

\subsection{Concentration result} \label{sec:process}
In this section we study a stochastic process consisting of rounds. In each round $i \geq 1$, a positive value $X_i \in (0, B]$ (for some parameter $B > 0$) arrives, and is flagged ``accepted'' with a probability $p \in [0, 1]$, independently, and ``rejected'' otherwise. The value $X_i$ itself might depend on the way previous values have been flagged, but not on the way $X_i$ itself is flagged. More formally, let $A$ be the set of indexes corresponding to accepted values, then $X_i$ is a function of the set $A \cap \{1, 2, \dotsc, i - 1\}$. The process terminates after $T \geq 1$ rounds, where $T$ itself might depend on the way values have been flagged. However, it is guaranteed that $T$ is upper bounded by a finite integer $\bar{T}$ and:
\[
	\sum_{i = 1}^T X_i
	\geq
	L
\]
for some parameter $L \geq 0$.

Let $\Sigma_A = \sum_{i \in A} X_i$ be the random sum of the accepted values. Our objective is to show a concentration bound for $\Sigma_A$. Let us first prove such a bound for the case when we make an additional assumption.

\begin{assumption} \label{asm:powers_of_2}
Each value $X_i$ is equal to $B/2^j$ for some value $j \geq 0$.
\end{assumption}

Using the above assumption, we can now define some additional notation. Let $\delta$ be the smallest number such that some value $X_i$ has a positive probability to take the value $\delta$. Observe that $\delta$ is well defined since the above process has only finitely many possible outcomes. By Assumption~\ref{asm:powers_of_2}, every value $X_i$ is a multiple of $\delta$.

It is helpful to think of the values $X_i$ as intervals placed one after the other on the axis of real numbers, starting from $0$. In other words, for every value $X_i$ we have an interval starting at $\sum_{j = 1}^{i - 1} X_j$ and ending at $\sum_{j = 1}^i X_j$. Taking this point of view, every interval $X_i$ can be partitioned into $X_i/\delta$ ranges of size $\delta$. Let us associate a random variable with each one of these ranges. More formally, for every $i \geq 1$, let $Y_i$ be a random variable taking the value $1$ when the range $(\delta(i - 1), \delta i)$ is contained within an accepted interval $X_j$, and the value $0$ in all other cases.

\begin{lemma} \label{lem:sum_Y}
Under Assumption~\ref{asm:powers_of_2}, $\Sigma_A = \delta \cdot \sum_{i = 1}^{\bar{T} \cdot B/\delta} Y_i$.
\end{lemma}
\begin{proof}
Fix a realization of the above process. Recall that $Y_i$ is zero whenever the range $(\delta(i-1), \delta i)$ is not contained within an accepted interval. On the other hand, consider an arbitrary accepted interval $X_i$. The interval $X_i$ contains $X_i / \delta$ ranges of size $\delta$. Moreover, all these ranges end at the point $\sum_{j = 1}^T X_j \leq \bar{T} \cdot B$ or earlier, and thus, their variables appear in the sum on the right hand side of the equality we want to prove. Hence, in conclusion, the contribution of $X_i$ to that sum is exactly $X_i / \delta$. The observation now follows since the intervals $\{X_i\}_{i = 1}^T$ are disjoint, and thus, so are their contributions to the sum.
\end{proof}

Let $I$ be the minimal integer such that $\delta I \geq L$.

\begin{observation} \label{obs:lower_bound}
Under Assumption~\ref{asm:powers_of_2}, $\Sigma_A \geq \delta \cdot \sum_{i = 1}^I Y_i$.
\end{observation}
\begin{proof}
Notice that $\bar{T} \cdot B$ is an upper bound on the sum $\sum_{j = 1}^T X_j$. On the other hand, $L$ is a lower bound on this sum, and thus, we get: $\bar{T} \cdot B \geq L$. Hence, $\delta(\bar{T} \cdot B / \delta) \geq L$. Since the term $\bar{T} \cdot B / \delta$ is an integer, the minimality of $I$ implies $I \leq \bar{T} \cdot B / \delta$. Using Lemma~\ref{lem:sum_Y} and the fact that the variables $Y_i$ are non-negative, we get:
\[
	\Sigma_A
	=
	\delta \cdot \sum_{i = 1}^{\bar{T} \cdot B/\delta} Y_i
	\geq
	\delta \cdot \sum_{i = 1}^I Y_i
	\enspace.
	\qedhere
\]
\end{proof}

Observation~\ref{obs:lower_bound} shows that it is enough for our purpose to prove a concentration bound for $\sum_{i = 1}^I Y_i$. The following observation gives another useful property of $I$.

\begin{observation} \label{obs:always_in_interval}
Under Assumption~\ref{asm:powers_of_2}, for every $1 \leq i \leq I$, the range $(\delta(i - 1), \delta i)$ is always contained within some interval $X_j$.
\end{observation}
\begin{proof}
Assume towards a contradiction that there is some realization of the process under which the range $(\delta(i - 1), \delta i)$ is not contained within some interval $X_j$. This implies:
\[
	L
	\leq
	\sum_{j = 1}^T X_j
	\leq
	\delta (i - 1)
	\leq
	\delta (I - 1)
	\enspace,
\]
contradicting the definition of $I$.
\end{proof}

Let us now study the distribution of the variables $\{Y_i\}_{i = 1}^I$. 

\begin{lemma} \label{lem:probability_Y}
Under Assumption~\ref{asm:powers_of_2}, $\Pr[Y_i = 1] = p$ for every $1 \leq i \leq I$. Hence, by linearity of expectation:
\[
	\bE\left[\sum_{i = 1}^I Y_i\right]
	=
	\sum_{i = 1}^I \bE\left[Y_i\right]
	=
	pI
	\enspace.
\]
\end{lemma}
\begin{proof}
For every $j \geq 1$, let $\cE_j$ be the event that $j \leq T$ and $(\delta(i - 1), \delta i)$ is included in the interval $X_j$. Observe that $\cE_j$ depends only the acceptance of intervals $X_{j'}$ for $j' < j$. Moreover, given that $X_j$ exists it is accepted with probability $p$, independently of the acceptance of previous intervals. Hence, we get:
\[
	\Pr[Y_i = 1 \mid \cE_j]
	=
	p
	\enspace.
\]
The observation now follows by the law of total probability since Observation~\ref{obs:always_in_interval} guarantees that $(\delta(i - 1), \delta i)$ is included in some interval, and thus, the event $\cE_j$ happens for exactly a single value of $j$.
\end{proof}

For every $1 \leq h \leq B / \delta$, let us define $V_h = \{1 \leq i \leq I \mid i \equiv h \pmod{B / \delta}\}$. Observe that the sets $\{V_h\}_{h = 1}^{B / \delta}$ form a disjoint partition of the indexes from $1$ to $I$.

\begin{observation} \label{obs:different_intervals}
Under Assumption~\ref{asm:powers_of_2}, for every $1 \leq h \leq B / \delta$ and two different indexes $i, i' \in V_h$, the ranges $(\delta(i - 1), \delta i)$ and $(\delta(i' - 1), \delta i')$ cannot be contained in one interval $X_j$.
\end{observation}
\begin{proof}
Assume without loss of generality that $i < i'$. The definition of $V_h$ guarantees that $i + B / \delta \leq i'$. Hence,
\[
	\delta i' - \delta(i - 1)
	=
	\delta(i' - i) + \delta
	\geq
	B + \delta
	\enspace.
\]
Hence, any interval $X_j$ containing both ranges $(\delta(i - 1), \delta i)$ and $(\delta(i' - 1), \delta i')$ must be of length at least $B + \delta$, which contradicts the definition of the process.
\end{proof}

\begin{lemma} \label{lem:independence}
Under Assumption~\ref{asm:powers_of_2}, for every $1 \leq h \leq B / \delta$, the variables of $\{Y_i \mid i \in V_h\}$ are independent.
\end{lemma}
\begin{proof}
For every $i \in V_h$, let $V_h^{<i}$ denote the intersection $V_h \cap \{1, 2, \dotsc, i - 1\}$. To prove the lemma it is enough to show that for every $i \in V_h$ the variable $Y_i$ takes the value $1$ with probability $p$ conditioned on any assignment to the variables of $\{Y_j \mid j \in V_h^{<i}\}$. Let $\cA$ denote an arbitrary such assignment having a non-zero probability. For every $1 \leq \ell \leq \bar{T}$, let $\cE_\ell$ be the event that the interval $X_\ell$ exists and contains the range $(\delta(i - 1), \delta i)$.

Assume $\cE_\ell$ happens. By Observation~\ref{obs:different_intervals} the variables of $\{Y_j \mid j \in V_h^{<i}\}$ correspond to ranges contained in intervals before $X_\ell$. Thus, the values of these variables only imply information about the acceptance of these intervals. Since the acceptance of $X_\ell$ is independent of the acceptance of previous intervals, we get that every $1 \leq \ell \leq \bar{T}$ obeying $\Pr[\cE_\ell \mid \cA] > 0$ must also obey:
\[
	\Pr[Y_i = 1 \mid \cE_\ell, \cA] = p \enspace.
\]

Clearly the events $\{\cE_\ell\}_{\ell = 1}^{\bar{T}}$ are disjoint. By Observation~\ref{obs:always_in_interval} we also know that one of them must happen. Hence, we get:
\[
	\Pr[Y_i \mid \cA]
	=
	\sum_{\substack{1 \leq \ell \leq \bar{T} \\ \Pr[\cE_\ell \mid \cA] > 0}} \Pr[\cE_\ell \mid \cA] \cdot \Pr[Y_i = 1 \mid \cE_\ell, \cA]
	=
	p \cdot \sum_{\ell = 1}^{\bar{T}} \Pr[\cE_\ell \mid \cA]
	=
	p
	\enspace.
	\qedhere
\]
\end{proof}

\begin{corollary} \label{cor:h_var_bound}
Under Assumption~\ref{asm:powers_of_2}, for every $1 \leq h \leq B / \delta$, $\Var\left[\sum_{i \in V_h} Y_i\right] \leq p(L/B + 2)$.
\end{corollary}
\begin{proof}
By Lemma~\ref{lem:probability_Y}, for every $i \in V_h$, $\Var[Y_i] = p(1 - p) \leq p$. Thus, by Lemma~\ref{lem:independence},
\[
	\Var\left[\sum_{i \in V_h} Y_i\right]
	=
	\sum_{i \in V_h} \Var\left[Y_i\right]
	\leq
	\sum_{i \in V_h} p
	=
	p \cdot |V_h|
	\enspace.
\]
The definition of $V_h$ guarantees that its size is at most:
\[
	|V_h|
	\leq
	\left\lceil \frac{I}{B / \delta} \right\rceil
	\leq
	\left\lceil \frac{L / \delta + 1}{B / \delta} \right\rceil
	\leq
	\frac{L + \delta}{B} + 1
	\leq
	\frac{L}{B} + 2
	\enspace.
	\qedhere
\]
\end{proof}

To bound the variance of the sum $\sum_{i = 1}^I Y_i$, we need the following simple technical lemma.

\begin{lemma} \label{lem:variance_of_sum}
For a set of random variables $Z_1, Z_2, \dotsc Z_\ell$, each having a finite variance,
\[
	\Var\left[\sum_{i = 1}^\ell Z_i\right]
	\leq
	\left(\sum_{i = 1}^\ell \sqrt{\Var[Z_i]}\right)^2
	\enspace.
\]
\end{lemma}
\begin{proof}
\[
	\Var\left[\sum_{i = 1}^\ell Z_i\right]
	=
	\sum_{i = 1}^\ell \sum_{j = 1}^\ell \Cov[Z_i, Z_j]
	\leq
	\sum_{i = 1}^\ell \sum_{j = 1}^\ell \sqrt{\Var[Z_i] \cdot \Var[Z_j]}
	=
	\left(\sum_{i = 1}^\ell \sqrt{\Var[Z_i]}\right)^2
	\enspace.
	\qedhere
\]
\end{proof}

\begin{corollary} \label{cor:total_variance}
Under Assumption~\ref{asm:powers_of_2}, $\Var\left[\sum_{i = 1}^I Y_i\right] \leq pB\delta^{-2}(L + 2B)$.
\end{corollary}
\begin{proof}
Observe that:
\[
	\sum_{i = 1}^I Y_i
	=
	\sum_{h = 1}^{B / \delta} \sum_{i \in V_h} Y_i
	\enspace.
\]
Hence, by Corollary~\ref{cor:h_var_bound} and Lemma~\ref{lem:variance_of_sum}:
\[
	\Var\left[\sum_{i = 1}^I Y_i\right]
	\leq
	\left(\frac{B}{\delta} \cdot \sqrt{p\left(\frac{L}{B} + 2 \right)}\right)^2
	=
	\frac{p B}{\delta^2}\left(L + 2B\right)
	\enspace.
	\qedhere
\]
\end{proof}

We are now ready to prove the promised concentration bound for $\Sigma_A$.

\begin{lemma} \label{lem:bound_process_with_assumption}
Under Assumption~\ref{asm:powers_of_2}, for every $t > 0$, $\Pr[\Sigma_A < pL - t] \leq pB(L + 2B)/t^2$.
\end{lemma}
\begin{proof}
By Lemma~\ref{obs:lower_bound},
\begin{align*}
	\Pr[\Sigma_A < pL - t]
	\leq{} &
	\Pr\left[\delta \cdot \sum_{i = 1}^I Y_i < pL - t\right]\\
	\leq{} &
	\Pr\left[\delta \cdot \sum_{i = 1}^I Y_i < \delta \cdot pI - t\right]
	\leq
	\Pr\left[\left|\delta \cdot \sum_{i = 1}^I Y_i - \delta \cdot pI\right| > t\right]
	\enspace.
\end{align*}
Since the expected value of $\delta \cdot \sum_{i = 1}^I Y_i$ is $\delta \cdot pI$ by Lemma~\ref{lem:probability_Y}, we get by Chebyshev's inequality:
\[
	\Pr\left[\left|\delta \cdot \sum_{i = 1}^I Y_i - \delta \cdot pI\right| > t\right]
	\leq
	\frac{\Var\left[\delta \cdot \sum_{i = 1}^I Y_i\right]}{t^2}
	=
	\frac{\delta^2 \cdot \Var\left[\sum_{i = 1}^I Y_i\right]}{t^2}
	\leq
	\frac{pB(L + 2B)}{t^2}
	\enspace,
\]
where the last inequality holds by Corollary~\ref{cor:total_variance}.
\end{proof}

Finally, we would like to get a version of Lemma~\ref{lem:bound_process_with_assumption} that holds without Assumption~\ref{asm:powers_of_2}.
\begin{corollary} \label{cor:bound_process}
For every $t > 0$, $\Pr[\Sigma_A < pL/2 - t] \leq pB(L + 2B)/(4t^2)$
\end{corollary}
\begin{proof}
For every value $X_i$, let us define a value $X'_i$ as follows:
\[
	X'_i = B / 2^{\lfloor \log_2(B / X_i) \rfloor}
	\enspace.
\]
Intuitively, $X'_i$ is the smallest value allowed by Assumption~\ref{asm:powers_of_2} that is at least as large as $X_i$. We say that $X'_i$ is accepted if and only if $X_i$ is. One can verify that the following holds:
\begin{compactitem}
	\item The values $X'_1, X'_2, \dotsc, X'_T$ define a legal process with the same parameters $p$, $B$ and $L$ as the original process, and this process obeys Assumption~\ref{asm:powers_of_2}. Let $\Sigma'_A$ be the sum of the accepted values in this process.
	\item For every value $X_i$, $X'_i \leq 2X_i$, hence, $\Sigma'_A \leq 2 \cdot \Sigma_A$.
\end{compactitem}
Thus, by Lemma~\ref{lem:bound_process_with_assumption}:
\[
	\Pr[\Sigma_A < pL/2 - t]
	\leq
	\Pr[\Sigma'_A < pL - 2t]
	\leq
	\frac{pB(L + 2B)}{(2t)^2}
	=
	\frac{pB(L + 2B)}{4t^2}
	\enspace.
	\qedhere
\]
\end{proof}

\subsection{Proof for $m^* \leq f(OPT)/(256(d + 1)^2)$} \label{ssc:proof_small_m}

In this section we prove that Algorithm~\ref{alg:multiple_elements_approximation} is $O(\beta)$-competitive when $m^* \leq f(OPT)/(256(d + 1)^2)$. Recall that Algorithm~\ref{alg:multiple_elements_approximation} consists of two parts. Let us say that Algorithm~\ref{alg:multiple_elements_approximation} ``applies the second option'' when it executes the second part (the one that involves the aided algorithm). Notice that it is enough to show that the algorithm is $O(\beta)$-competitive when it applies the second option, since this option is applied with probability $1/2$.

We need some additional notation. First, we denote by $\ell + 1$ the number of iterations performed by the loop on Line~\ref{line:loop} of the algorithm (\ie, the $\ell + 1$ iteration is the iteration at which the algorithm decides to leave the loop, and does not change $A$). Additionally, for every $1 \leq i \leq \ell$, let $A_i$ and $W_i$ denote the set $A$ and the value $W$, respectively, immediately after the $i$-th iteration of this loop. For consistency, we also denote by $A_0$ and $W_0$ these set and value before the first iteration. Finally, for every $1 \leq i \leq \ell$, let $u_i$ denote the element $u$ chosen at iteration $i$ of the loop. We begin the analysis of the small $m^*$ case by proving an upper bound on $W_\ell$ (the final value of $W$).

\begin{observation} \label{obs:independence}
When Algorithm~\ref{alg:multiple_elements_approximation} applies the second option, $A_i \in \cI$ for every $0 \leq i \leq \ell$.
\end{observation}
\begin{proof}
The observation holds since Algorithm~\ref{alg:multiple_elements_approximation} chooses at every iteration an element $u$ and a set $\cD_u$ whose addition to $A$ does not violate independence.
\end{proof}

\begin{lemma} \label{lem:W_upper_bound}
When Algorithm~\ref{alg:multiple_elements_approximation} applies the second option, $W_\ell \leq f(OPT)$.
\end{lemma}
\begin{proof}
We prove by induction on $i$ the claim that the inequality $W_i \leq f(A_i)$ holds for every $0 \leq i \leq \ell$. Notice that the observation follows from this claim since, by Observation~\ref{obs:independence}, $A_\ell$ is independent, and thus, $f(A_\ell) \leq f(OPT)$.

For $i = 0$ the claim is trivial since $W_0 = 0 = f(\varnothing) = f(A_0)$. For $i > 0$, assume the claim holds for $i - 1$, and let us prove it for $i$.
\begin{align*}
	W_i
	={} &
	W_{i - 1} + f(u_i \mid A_{i - 1} \cup \cD_{u_i})
	\leq
	f(A_{i - 1}) + f(u_i \mid A_{i - 1} \cup \cD_{u_i})\\
	\leq{} &
	f(A_{i - 1} \cup \cD_{u_i}) + f(u_i \mid A_{i - 1} \cup \cD_{u_i})
	=
	f(A_i)
	\enspace,
\end{align*}
where the first inequality holds by the induction hypothesis, and the second inequality follows from the monotonicity of $f$.
\end{proof}

Our next objective is to prove a lower bound on $W_\ell$ that holds with a constant probability. For that purpose, let us consider an offline algorithm which calculates a value $W$ having the same distribution as the value $W$ calculated by Algorithm~\ref{alg:multiple_elements_approximation}.

\begin{algorithm}[H]
\DontPrintSemicolon
\caption{\textsf{Offline W Calculation}}
Let $A \gets \varnothing$, $W \gets 0$ and $T \gets \cN$.\\
\While{there exist $u \in T \setminus A$ and $\cD_u \subseteq \DS{u}$ s.t. $A \cup \cD_u + u \in \cI$}
{
	Find such a pair maximizing $f(u \mid A \cup \cD_u)$.\\	
	
	\With{probability $(d + 2)^{-1}$}
	{
		Increase $W \gets W + f(u \mid A \cup \cD_u)$.\\
		Update $A \gets A \cup \cD_u + u$.
	}
	\lOtherwise{Update $T \gets T - u$.}
	
}
\end{algorithm}

\begin{observation} \label{obs:same_distribution}
The distribution of the value $W$ calculated by Algorithm~\ref{alg:offline_W} is identical to the distribution of $W_\ell$ as calculated by Algorithm~\ref{alg:multiple_elements_approximation} when it applies the second option.
\end{observation}
\begin{proof}
The loop starting on Line~\ref{line:loop} of Algorithm~\ref{alg:multiple_elements_approximation} looks in each iteration for a pair of an element $u$ and a set $\cD_u$ maximizing $f(u \mid A \cup \cD_u)$ and obeying some other conditions, including the requirement that $u$ belongs to a set $T$ containing every element with probability $(d + 2)^{-1}$, independently.

On the other hand, Algorithm~\ref{alg:offline_W} has a loop that looks for a pair of an element $u$ and a set $\cD_u$ maximizing $f(u \mid A \cup \cD_u)$ and obeying the same conditions, except for requiring $u$ to belong to $T$. Once such a pair is found, the algorithm makes a random decision whether to keep $u$ in $T$, and then uses the pair to increase the solution $A$ if and only if it decides to keep $u$ in $T$.

Notice that the only difference between these two procedures is the point when the algorithm decides whether each element $u$ should belong to $T$. Algorithm~\ref{alg:multiple_elements_approximation} makes the decisions for all elements at the beginning, while Algorithm~\ref{alg:offline_W} makes the decisions only when necessary. However, regardless of when the membership of elements in $T$ is decided, the set of pairs used to increase the solution $A$ is the same given that the same random decisions are made by both algorithms.
\end{proof}

To analyze the distribution of the value $W$ calculated by Algorithm~\ref{alg:offline_W} we need some additional notation. Let $\hat{\ell} + 1$ be the number of iterations performed by the loop of Algorithm~\ref{alg:offline_W} (\ie, $\hat{\ell}$ is the number of times elements are added to the set $A$), and for every $1 \leq i \leq \hat{\ell}$ let $\hat{A}_i$, $\hat{T}_i$ and $\hat{W}_i$ denote the sets $A$ and $T$ and the value $W$, respectively, immediately after the $i$-th iteration of this loop. For consistency, we also denote by $\hat{A}_0$, $\hat{T}_0$ and $\hat{W}_0$ these sets and value before the first iteration. Additionally, for every $1 \leq i \leq \hat{\ell}$, let $\hat{u}_i$ denote the element $u$ chosen at iteration $i$ of the loop. Finally, for every $0 \leq i \leq \hat{\ell}$, let us denote by $OPT_i$ the maximum value independent set that can be obtained from $\hat{A}_i$ by adding only $T$ elements. Formally,
\[
	OPT_i = \operatorname*{\arg \max}_{S \in \cI \mid \hat{A}_i \subseteq S \subseteq \hat{A}_i \cup \hat{T}_i} f(S) \enspace.
\]
Observe that $OPT_i$ is well defined since the set $\hat{A}_i$ is independent for every $0 \leq i \leq \hat{\ell}$.

Next, observe that Algorithm~\ref{alg:offline_W} can be viewed as a process of the kind described in Section~\ref{sec:process}. More precisely, each iteration $i$ of Algorithm~\ref{alg:offline_W} corresponds to one iteration of the process, the value of this iteration is $f(\hat{u}_i \mid \hat{A}_{i - 1} \cup \cD_{\hat{u}_i})$ and this value is accepted if and only if $\hat{u}_i$ is kept in $T$ (and $f(\hat{u}_i \mid \hat{A}_{i - 1} \cup \cD_{\hat{u}_i})$ is added to $W$). Note that, as required by the process definition, the value $f(\hat{u}_i \mid \hat{A}_{i - 1} \cup \cD_{\hat{u}_i})$ depends only on the acceptances of previous values, and the acceptances of the value $f(\hat{u}_i \mid \hat{A}_{i - 1} \cup \cD_{\hat{u}_i})$ is independent of anything else. Taking this point of view, $W_{\hat{\ell}}$ is exactly the sum of the accepted values, and thus, can be analyzed using Corollary~\ref{cor:bound_process}. To use the last corollary, we need to determine the parameters of the process:
\begin{compactitem}
	\item By definition, $m^*$ upper bounds $f(\hat{u}_i \mid (\hat{A}_{i - 1} \cup \cD_{\hat{u}_i}) \cap \DS{u}) \geq f(\hat{u}_i \mid \hat{A}_{i - 1} \cup \cD_{\hat{u}_i})$. Hence, one can choose $B = f(OPT) / [256(d + 1)^2] \geq m^*$ for the process given by Algorithm~\ref{alg:offline_W}.
	\item Every value is accepted with probability $(d + 2)^{-1}$, thus, this is the value of $p$.
\end{compactitem}

We are left to determine a possible value for the parameter $L$. For that purpose we need a few claims.
\begin{observation} \label{obs:equality}
$OPT_{\hat{\ell}} = \hat{A}_{\hat{\ell}}$.
\end{observation}
\begin{proof}
By definition $OPT_{\hat{\ell}}$ contains the elements of $\hat{A}_{\hat{\ell}}$ and (possibly) additional elements of $\hat{T}_{\hat{\ell}}$ that do not violate independence when added to $\hat{A}_{\hat{\ell}}$. However, the fact that Algorithm~\ref{alg:offline_W} stopped during the $\hat{\ell} + 1$ iteration implies that no elements of $\hat{T}_{\hat{\ell}}$ can be added to $\hat{A}_{\hat{\ell}}$ without violating independence. Therefore, $OPT_{\hat{\ell}}$ cannot contain any elements beside the elements of $\hat{A}_{\hat{\ell}}$.
\end{proof}

For every $0 \leq i \leq \hat{\ell}$, let $L_i$ be the sum of the first $i$ values of the process corresponding to Algorithm~\ref{alg:offline_W}. Formally,
\[
	L_i = \sum_{j = 1}^i f(\hat{u}_j \mid \hat{A}_{j - 1} \cup \cD_{\hat{u}_j})
	\enspace.
\]

\begin{lemma} \label{lem:alg_gain}
For every $0 \leq i \leq \hat{\ell}$, $f(\hat{A}_i) \leq (d + 1) \cdot L_i$.
\end{lemma}
\begin{proof}
We prove by induction on $i$ that for every $0 \leq i \leq \hat{\ell}$:
\begin{equation} \label{eq:induction}
	f(\hat{A}_i)
	\leq
	(d + 1) \cdot \sum_{\hat{u}_j \in \hat{A}_i} f(\hat{u}_j \mid \hat{A}_{j - 1} \cup \cD_{\hat{u}_j})
	\enspace.
\end{equation}
Observe that the lemma follows from the last claim since $\hat{u}_j$ can belong to the set $\hat{A}_i$ only when $j \leq i$. For $i = 0$, Equation~\eqref{eq:induction} is trivial since $f(\hat{A}_0) = f(\varnothing) = 0$. Next, assume Equation~\eqref{eq:induction} holds for $i - 1 \geq 0$, and let us prove it for $i$. If $\hat{u}_i$ is removed from $T$ then this is true since in this case $\hat{A}_i = \hat{A}_{i - 1}$. Thus, we can safely assume in the rest of the proof that $\hat{u}_i$ remains in $T$.

For every $0 \leq i \leq \hat{\ell}$, let $N_i = \{\hat{u}_j \not \in A_i \mid 1 \leq j \leq i\}$. Order the elements of $\hat{A}_i \setminus \hat{A}_{i - 1}$ in an arbitrary order, and let $v_j$ denote the $j$-th element in this order. Consider some $1 \leq j \leq |\hat{A}_i \setminus \hat{A}_{i - 1}|$, if $v_j \not \in N_{i - 1}$ then the pair of the element $v_j$ and the set $\{v_1, v_2, \dotsc, v_{i - 1}\} \cap \DS{v_j}$ form a possible pair that Algorithm~\ref{alg:offline_W} could select on iteration $i$ since $v_j \in \cN \setminus (N_{i - 1} \cup \hat{A}_{i - 1}) = \hat{T}_{i - 1} \setminus \hat{A}_{i - 1}$. Hence,
\begin{align*}
	f(\hat{u}_i \mid \hat{A}_{i - 1} \cup \cD_{\hat{u}_i})
	\geq{} &
	f(v_j \mid \hat{A}_{i - 1} \cup (\{v_1, v_2, \dotsc, v_{i - 1}\} \cap \DS{v_j}))\\
	\geq{} &
	f(v_j \mid \hat{A}_{i - 1} \cup \{v_1, v_2, \dotsc, v_{i - 1}\})
	\enspace.
\end{align*}
On the other hand, if $v_j \in N_{i - 1}$, then there must be some $1 \leq h < i$ such that $v_j = \hat{u}_h$, and thus, $v_j \in \hat{T}_h$. By a similar argument to the one used above we get:
\begin{align*}
	f(\hat{u}_h \mid \hat{A}_{h - 1} \cup \cD_{\hat{u}_h})
	\geq{} &
	f(v_j \mid \hat{A}_{h - 1} \cup ((\hat{A}_{i - 1} \cup \{v_1, v_2, \dotsc, v_{i - 1}\}) \cap \DS{v_j}))\\
	\geq{} &
	f(v_j \mid \hat{A}_{h - 1} \cup (\hat{A}_{i - 1} \cup \{v_1, v_2, \dotsc, v_{i - 1}\}))\\
	={} &
	f(v_j \mid \hat{A}_{i - 1} \cup \{v_1, v_2, \dotsc, v_{i - 1}\})
	\enspace.
\end{align*}

Adding the inequalities we got for every $1 \leq j \leq |\hat{A}_i \setminus \hat{A}_{i - 1}|$ gives:
\begin{align*}
	f(\hat{A}_i) -& f(\hat{A}_{i - 1})
	=
	\sum_{j = 1}^{|\hat{A}_i \setminus \hat{A}_{i - 1}|} f(v_j \mid \hat{A}_{i - 1} \cup \{v_1, v_2, \dotsc, v_{i - 1}\})\\
	={} &
	\sum_{v_j \in \hat{A}_i \setminus (\hat{A}_{i - 1} \cup N_{i - 1})} \mspace{-36mu} f(v_j \mid \hat{A}_{i - 1} \cup \{v_1, v_2, \dotsc, v_{i - 1}\}) + \sum_{v_j \in \hat{A}_i \cap N_{i - 1}} \mspace{-18mu} f(v_j \mid \hat{A}_{i - 1} \cup \{v_1, v_2, \dotsc, v_{i - 1}\})\\
	\leq{} &
	|\hat{A}_i \setminus (\hat{A}_{i - 1} \cup N_{i - 1})| \cdot f(\hat{u}_i \mid \hat{A}_{i - 1} \cup \cD_{\hat{u}_i}) + \sum_{u_h \in \hat{A}_i \cap N_{i - 1}} \mspace{-18mu} f(\hat{u}_h \mid \hat{A}_{h - 1} \cup \cD_{\hat{u}_h})
	\enspace.
\end{align*}
Observe that $|\hat{A}_i \setminus (\hat{A}_{i - 1} \cup N_{i - 1})| \leq d + 1$ and $\hat{A}_i \cap N_{i - 1} \subseteq \hat{A}_i \setminus \hat{A}_{i - 1} - \hat{u}_i$ . Combining both observations with the last inequality yields:
\begin{align*}
	f(\hat{A}_i)
	\leq{} &
	f(\hat{A}_{i - 1}) + (d + 1) \cdot f(\hat{u}_i \mid \hat{A}_{i - 1} \cup \cD_{\hat{u}_j}) + \sum_{u_h \in \hat{A}_i \setminus \hat{A}_{i - 1} - \hat{u}_i} \mspace{-27mu} f(\hat{u}_h \mid \hat{A}_{h - 1} \cup \cD_{\hat{u}_h}) \\
	\leq{} &
	f(\hat{A}_{i - 1}) + (d + 1) \cdot \sum_{u_h \in \hat{A}_i \setminus \hat{A}_{i - 1}} \mspace{-18mu} f(\hat{u}_h \mid \hat{A}_{h - 1} \cup \cD_{\hat{u}_h})
	\leq
	(d + 1) \cdot \sum_{u_h \in \hat{A}_i} f(\hat{u}_h \mid \hat{A}_{h - 1} \cup \cD_{\hat{u}_h})
	\enspace,
\end{align*}
where the last inequality holds by the induction hypothesis.
\end{proof}

\begin{lemma} \label{lem:opt_loss}
For every $0 \leq i \leq \hat{\ell}$, $f(OPT_i) + (d + 1) \cdot L_i \geq f(OPT)$.
\end{lemma}
\begin{proof}
We prove the lemma by induction on $i$. For $i = 0$ the lemma is trivial since $f(OPT_0) = f(OPT)$. Next, assume that the lemma holds for $i - 1 \geq 0$, and let us prove it for $i$. There are two cases to consider. First, let us consider the case where $\hat{u}_i$ is removed from $T$ and $\hat{A}_i = \hat{A}_{i - 1}$. In this case one potential candidate for $OPT_i$ is $OPT_{i - 1} - \hat{u}_i$. If $\hat{u}_i \not \in OPT_{i - 1}$, then we get $f(OPT_{i - 1} - \hat{u}_i) = f(OPT_{i - 1})$. On the other hand, if $\hat{u}_i \in OPT_{i - 1}$ then the pair of the element $\hat{u}_i$ and the set $(OPT_{i - 1} \setminus \hat{A}_{i - 1}) \cap \DS{\hat{u}_i}$ is a possible pair that Algorithm~\ref{alg:offline_W} could select on iteration $i$. Hence,
\begin{align*}
	f(OPT_{i - 1} - \hat{u}_i)
	={} &
	f(OPT_{i - 1}) - f(\hat{u}_i \mid \hat{A}_{i - 1} \cup (OPT_{i - 1} \setminus \hat{A}_{i - 1} - \hat{u}_i))\\
	\geq{} &
	f(OPT_{i - 1}) - f(\hat{u}_i \mid \hat{A}_{i - 1} \cup ((OPT_{i - 1} \setminus \hat{A}_{i - 1}) \cap \DS{\hat{u}_i}))\\
	\geq{} &
	f(OPT_{i - 1}) - f(\hat{u}_i \mid \hat{A}_{i - 1} \cup \cD_{\hat{u}_i})
	\enspace.
\end{align*}
Therefore, regardless of the membership of $\hat{u}_i$ in $OPT_{i -1}$, we can lower bound by $f(OPT_{i - 1} - \hat{u}_i)$ by $f(OPT_{i - 1}) - f(\hat{u}_i \mid \hat{A}_{i - 1} \cup \cD_{\hat{u}_i})$. Using the induction hypothesis, we now get:
\begin{align*}
	f(OPT_i) + (d + 1) \cdot L_i
	\geq{} &
	f(OPT_{i - 1} - \hat{u}_i) + (d + 1) \cdot L_{i - 1} + (d + 1) \cdot f(\hat{u}_i \mid \hat{A}_{i - 1} \cup \cD_{\hat{u}_i})\\
	\geq{} &
	f(OPT_{i - 1}) + (d + 1) \cdot L_{i - 1}
	\geq
	f(OPT)
	\enspace.
\end{align*}

Next, let us consider the case where $\hat{u}_i$ is kept in $T$ and $\hat{A}_i = \hat{A}_{i - 1} \cup \cD_{\hat{u}_i} + \hat{u}_i$. In this case, by standard matroid properties, one can obtain a candidate for $OPT_i$ by starting with $OPT_{i - 1} \cup \cD_{\hat{u}_i} + \hat{u}_i$ and removing from it a subset $\Delta \subseteq OPT_{i - 1} \setminus \hat{A}_i$ of up to $d + 1$ elements. Let us denote by $OPT'_i$ this candidate, \ie, $OPT'_i = (OPT_{i - 1} \cup \cD_{\hat{u}_i} + \hat{u}_i) \setminus \Delta$.

Order the elements of $\Delta$ in an arbitrary order, and let $v_j$ denote the $j$-th element in this order. Observe that by monotonicity:
\begin{align*}
	f(OP&T_{i - 1}) - f(OPT'_i)
	\leq
	f(OPT_{i - 1}) - f(OPT_{i - 1} \setminus \Delta)
	=
	\sum_{j = 1}^{|\Delta|} f(v_j \mid OPT_{i - 1} \setminus \{v_1, v_2, \dotsc, v_j\})\\
	\leq{} &
	\sum_{j = 1}^{|\Delta|} f(v_j \mid A_{i - 1} \cup ((OPT_{i - 1} \setminus \{v_1, v_2, \dotsc, v_j\}) \cap \DS{v_j}))
	\leq
	(d + 1) \cdot f(\hat{u}_i \mid \hat{A}_{i - 1} \cup \DS{\hat{u}_i})
	\enspace.
\end{align*}
where the last inequality holds since $v_j \in OPT_{i - 1} \setminus \hat{A}_i \subseteq \hat{T}_{i - 1}$ for every $1 \leq j \leq |\Delta|$, and thus, the pair of the element $v_j$ and the set $[(OPT_{i - 1} \setminus \{v_1, v_2, \dotsc, v_j\}) \setminus \hat{A}_{i - 1}] \cap \DS{v_j} $ is a possible pair that Algorithm~\ref{alg:offline_W} could select on iteration $i$. Using the induction hypothesis, we now get:
\begin{align*}
	f(OPT_i) + (d + 1) \cdot L_i
	\geq{} &
	f(OPT'_i) + (d + 1) \cdot L_{i - 1} + (d + 1) \cdot f(\hat{u}_i \mid \hat{A}_{i - 1} \cup \cD_{\hat{u}_i})\\
	\geq{} &
	f(OPT_{i - 1}) + (d + 1) \cdot L_{i - 1}
	\geq
	f(OPT)
	\enspace.
	\qedhere
\end{align*}
\end{proof}

\begin{corollary}
The parameter $L$ of the process corresponding to Algorithm~\ref{alg:offline_W} can be chosen to be $f(OPT) / [2(d + 1)]$.
\end{corollary}
\begin{proof}
Observe that any value that always lower bounds $L_{\hat{\ell}}$ can be used as a value for the parameter $L$. Combining Lemmata~\ref{lem:alg_gain} and~\ref{lem:opt_loss} gives:
\[
	2(d + 1) \cdot L_{\hat{\ell}}
	\geq
	f(\hat{A}_{\hat{\ell}}) + [f(OPT) - f(OPT_{\hat{\ell}})]
	=
	f(OPT)
	\enspace,
\]
where the equality holds by Observation~\ref{obs:equality}.
\end{proof}

Now that we have all the parameters of the process corresponding to Algorithm~\ref{alg:offline_W}, we can use Corollary~\ref{cor:bound_process} to give a guarantee on $W_\ell$.

\begin{lemma} \label{lem:W_lower_bound}
When Algorithm~\ref{alg:multiple_elements_approximation} applies the second option and $m^* \leq f(OPT)/[256(d + 1)^2]$, then, with probability at least $7/8$, $W_\ell \geq \frac{f(OPT)}{8(d + 2)^2}$.
\end{lemma}
\begin{proof}
Recall that $\hat{W}_{\hat{\ell}}$ is the sum of the accepted values in the process corresponding to Algorithm~\ref{alg:offline_W}. Hence, by Corollary~\ref{cor:bound_process},
\begin{align*}
	\Pr\left[\hat{W}_{\hat{\ell}} < \frac{f(OPT)}{8(d + 2)^2}\right]
	\leq{} &
	\Pr\left[\hat{W}_{\hat{\ell}} < \frac{pL}{2} - \frac{pL}{4}\right]
	\leq
	\frac{pB(L + 2B)}{4(pL/4)^2}
	=
	\frac{4B(L + 2B)}{pL^2}\\
	={} &
	\frac{4 \cdot \frac{f(OPT)}{256(d + 1)^2} \cdot \left(\frac{f(OPT)}{2(d + 1)} + \frac{2 \cdot f(OPT)}{256(d + 1)^2}\right)}{\frac{1}{d + 2} \cdot \left(\frac{f(OPT)}{2(d + 1)}\right)^2}
	\leq
	\frac{4 \cdot \frac{1}{256(d + 1)^2} \cdot \frac{1}{d + 1}}{\frac{1}{d + 2} \cdot \frac{1}{4(d + 1)^2}}
	=
	\frac{(d + 2)}{16(d + 1)}
	\leq
	\frac{1}{8}
	\enspace.
\end{align*}
The lemma now follows since $W_\ell $ and $\hat{W}_{\hat{\ell}}$ have the same distribution by Observation~\ref{obs:same_distribution}.
\end{proof}

To complete the analysis of Algorithm~\ref{alg:multiple_elements_approximation} we also need the following notation and lemma.
Given a set $S \subseteq \cN$, let $S(p)$ be a random set containing every element $u \in S$, independently, with probability $p$.
\begin{lemma} \label{lem:sampling}
For every set $S \subseteq \cN$, $\bE[f(S(p))] \geq p^{d + 1} \cdot f(S)$.
\end{lemma}

\begin{proof}
For every element $u \in S$, let $X_u$ be an indicator for the event that $S \cap \DS{u} + u \in S(p)$. Clearly, $\Pr[X_u = 1] \geq p^{d + 1}$. Also, let $u_1, u_2, \dotsc, u_{|S|}$ denote an arbitrary order of the elements of $S$. Then,
\begin{align*}
	\mathbb{E}[f(S(p))]
	={} &
	\sum_{i = 1}^{|S|} \bE[f(\{u_i\} \cap S(p) \mid S(p) \cap \{u_1, u_2, \dotsc, u_{i - 1}\})]\\
	\geq{} &
	\sum_{i = 1}^{|S|} \bE[X_i \cdot f(u_i \mid \{u_1, u_2, \dotsc, u_{i - 1}\})]\\
	\geq{} &
	p^{d + 1} \cdot \sum_{i = 1}^{|S|} f(u_i \mid \{u_1, u_2, \dotsc, u_{i - 1}\})
	=
	p^{d + 1} \cdot f(S)
	\enspace,
\end{align*}
where the first inequality follows from the definition of $\DS{u}$.
\end{proof}

\begin{lemma} \label{lem:no_big_multiple}
If $m^* \leq f(OPT)/[256(d + 1)^2]$, then Algorithm~\ref{alg:multiple_elements_approximation} is $O(\beta)$-competitive.
\end{lemma}
\begin{proof}
Throughout this proof we assume that Algorithm~\ref{alg:multiple_elements_approximation} applies the second option. This event happens with probability $1/2$, thus, it is enough to prove that Algorithm~\ref{alg:multiple_elements_approximation} is $O(\beta)$-competitive given this event.

Let $R = \cN \setminus T_{\ell}$ be the set of elements that are not observed by Algorithm~\ref{alg:multiple_elements_approximation}. By Lemma~\ref{lem:sampling}:
\begin{align*}
	\bE[f(OPT \cap R)]
	={} &
	\bE\left[f\left(OPT\left(1 - \frac{1}{d + 2}\right)\right)\right]\\
	\geq{} &
	\left(1 - \frac{1}{d + 2}\right)^{d + 1} \cdot f(OPT)
	\geq
	e^{-1} \cdot f(OPT)
	\enspace.
\end{align*}
Hence,
\[
	\bE[f(OPT(R))]
	\geq
	\bE[f(OPT \cap R)]
	\geq
	\frac{f(OPT)}{e}
	\geq
	\frac{f(OPT)}{4}
	\enspace,
\]
where $OPT(R)$ is the independent subset of $R$ maximizing $f$. Since $f(OPT(R))$ is always upper bounded by $f(OPT)$, this implies the following claim:
\[
	\Pr\left[f(OPT(R)) \geq \frac{f(OPT)}{10}\right]
	\geq
	\frac{1}{6}
	\enspace.
\]

On the other hand, by Lemma~\ref{lem:W_lower_bound}, $W_\ell \geq f(OPT) / [8(d + 2)^2]$ with probability at least $7/8$. Hence, by the union bound we have:
\[
	W_\ell \geq \frac{f(OPT)}{8(d + 2)^2}
	\qquad
	\text{and}
	\qquad
	f(OPT(R)) \geq \frac{f(OPT)}{10}
\]
with probability at least $1/24$. To complete the proof it is enough to show that $ALG$ is $O(\beta)$-competitive when the last two inequalities hold. This follows from the definition of $ALG$ when $\opt_{80(d + 2)^2}$ is a valid approximation for $f(OPT(R))$, \ie, when we have $f(OPT(R)) / [80(d + 2)^2] \leq \opt_{80(d + 2)^2} \leq f(OPT(R))$. Thus, in the rest of the proof we prove these inequalities:
\[
	\opt_{80(d + 2)^2}
	=
	\frac{W_\ell}{10}
	\leq
	\frac{f(OPT)}{10}
	\leq
	f(OPT(R))
	\enspace,
\]
where the first inequality follows from Lemma~\ref{lem:W_upper_bound}. On the other hand,
\[
	f(OPT(R))
	\leq
	f(OPT)
	\leq
	8(d + 2)^2W_\ell
	=
	80(d + 2) \cdot \opt_{80(d + 2)^2}
	\enspace.
	\qedhere
\]
\end{proof}

Note that Theorem~\ref{thm:multiple_elements_approximation} follows immediately from Lemmata~\ref{lem:big_multiple} and~\ref{lem:no_big_multiple}.

\end{document}